\documentclass[12pt]{iopart}
\eqnobysec
\usepackage[mathscr]{euscript}
\expandafter\let\csname equation*\endcsname\relax
\expandafter\let\csname endequation*\endcsname\relax
\usepackage{amsmath, amsthm, amssymb}
\usepackage{enumerate}
\usepackage{setspace}
\usepackage{natbib}

\newtheorem{theorem}{Theorem}[section]

\newtheorem{proposition}[theorem]{Proposition}

\newenvironment{remark}[1][Remark]{\begin{trivlist}
\item[\hskip \labelsep {\bfseries #1}]}{\end{trivlist}}

\newcommand{\derv}[1]{\frac{\partial}{\partial #1}}

\newcommand{\deriv}[2]{\frac{\partial #1}{\partial #2}}

\newcommand{\beqn}{\begin{equation}}
\newcommand{\eeqn}{\end{equation}}
\newcommand{\beqnar}{\begin{eqnarray}}
\newcommand{\eeqnar}{\end{eqnarray}}

\newcommand{\contr}{\,\lrcorner\,}
\usepackage{hyperref}
\begin{document}

\title[Vorticity and Symplecticity in Lagrangian Gas Dynamics] 
{Vorticity  and Symplecticity in Multi-Symplectic, Lagrangian Gas Dynamics} 
\author{G.M. Webb${}^1$ and S.C. Anco${}^2$} 
\address{${}^1$ Center for Space Plasma and Aeronomic Research, The University of Alabama in Huntsville, 
 Huntsville AL 35805, USA}

\address{${}^2$Department of Mathematics, Brock University, 
St. Catherines, ON L2S 3A1 Canada}

 
\ead{gmw0002@uah.edu}


\begin{abstract}

The Lagrangian, multi-dimensional, ideal, compressible  gasdynamic equations are written 
in a multi-symplectic form, in which the Lagrangian fluid labels, $m^i$ (the Lagrangian
mass coordinates) and time $t$ are the independent variables, and in which the 
Eulerian position of the fluid element ${\bf x}={\bf x}({\bf m},t)$ 
and the entropy $S=S({\bf m},t)$ are the dependent variables. Constraints in the 
variational principle are incorporated by means of Lagrange multipliers. 
The constraints are: the entropy 
advection equation $S_t=0$, the Lagrangian map equation ${\bf x}_t={\bf u}$ 
where ${\bf u}$ is the fluid velocity, and the mass continuity equation which has the form 
$J=\tau$ where $J=\det(x_{ij})$ is the Jacobian of the Lagrangian map in which 
$x_{ij}=\partial x^i/\partial m^j$ and $\tau=1/\rho$ is the specific volume 
of the gas. 
The internal energy per unit volume of the gas $\varepsilon=\varepsilon(\rho,S)$ 
corresponds to a non-barotropic gas. 
The Lagrangian is used to define multi-momenta, 
and to develop  de-Donder Weyl Hamiltonian equations. The de Donder Weyl 
equations are cast in  a multi-symplectic form.
The pullback conservation laws and the symplecticity conservation laws  
are obtained.   
One class of symplecticity conservation laws give rise 
to vorticity and potential vorticity type conservation laws, and another class of symplecticity 
laws are related to derivatives of the Lagrangian energy conservation law with respect 
to the Lagrangian mass coordinates $m^i$. We show that the vorticity-symplecticity  
laws can be derived by a Lie dragging method, and also by using Noether's second theorem 
and a fluid relabelling symmetry  which is a divergence symmetry of the action.    
We obtain the Cartan-Poincar\'e form describing the equations  
and we discuss a set of differential forms representing the equation system.  

\end{abstract}

\noindent{{\it J. Phys. A., Math. and Theor.}, {\bf 49} (2016) 075501, doi:10.1088/1751-8133/49/7/075501}

\maketitle


\section{Introduction}
Multi-momentum, Hamiltonian systems were developed by de Donder (1930) and Weyl (1935). 
They obtained a generalization of Hamiltonian mechanics by using multi-momentum 
maps, in which there can be more than one generalized momentum 
corresponding to each canonical coordinate $q$. In this approach time $t$ in a fixed 
reference frame is not the only evolution variable in the equations (e.g. the 
system can also be thought of as evolving in the space variables). The de Donder Weyl 
Hamiltonian equations apply to action principles in which the Lagrangian 
$L=L({\bf x},\varphi^i,\partial\varphi^i/\partial x^\mu)$ where ${\bf x}$ are the independent 
variables and the $\varphi^k$ are the dependent variables ($1\leq k\leq m$ say, $k$ integer), 
which includes at least two independent partial derivatives $\partial\varphi^k/\partial x^s$, 
($1\leq s\leq n$, $n\geq 2$). 
Webb (2015) cast  the equations of ideal, 1D, Lagrangian gas dynamics in the 
de Donder-Weyl Hamiltonian form. In this development,  
the dependent variables are the Eulerian particle position $x=x(m,t)$ 
and gas entropy $S=S(m,t)$ where $S_t=0$. The multi-momenta for the system are 
the variables: $\pi_x^t=\partial L/\partial x_t$ and $\pi_x^m=\partial L/\partial x_m$ 
and $\pi_S^t=\partial L/\partial S_t$ where $L$ is the Lagrangian for the system. 
The de Donder-Weyl Hamiltonian equations, obtained by using a generalized Legendre transformation,
  were also cast in multi-symplectic form (see e.g. Hydon (2005) for a clear description 
of multi-symplectic systems of differential equations). 

There is an extensive literature on multi-momentum and multi-symplectic systems (e.g. 
Kanatchikov (1993,1997,1998), Forger et al. (2003), Forger and Gomes (2013), 
Forger and Romero (2005), Forger and Salles (2015), 
Gotay (1991a,b),
Gotay et al. (2004a,b), Roman Roy (2009), Marsden et al. (1986), Marsden and Shkoller (1999), 
Carenina et al. (1991),
Bridges et al. (2005,2010) and Cantrijn et al. (1999)). 

Anco and Dar (2009) carry out a Lie symmetry 
analysis and classification of conservation laws of compressible isentropic flow in $n>1$ 
spatial dimensions. Anco and Dar (2010) extended their (2009) analysis to the case 
of non-isentropic, inviscid flow in $n>1$ spatial dimensions. They give both the symmetries and conservation 
laws due to the ten Galilean point symmetries of the equations, as well as conservation laws
associated with helicity and vorticity. Anco, Dar and Tufail (2015) generalize their symmetry analysis 
to determine conserved integrals for inviscid, compressible fluid flow in Riemannian manifolds, 
for moving domains, in which Killing's equations and curl free homothetic Killing vectors 
play an important role. Cheviakov (2014) has derived new conservation laws for fluid systems involving 
vorticity and vorticity related equations (potential type systems)
including magnetohydrodynamics (MHD) and Maxwell's equations. Cheviakov and Oberlack (2014) 
derive generalized Ertel's theorems and infinite heirarchies of conserved 
quantities for the Euler and Navier Stokes equations. Kelbin et al. (2013) 
obtain new conservation laws in 
helically symmetric, plane and rotationally symmetric flows. Webb et al. (2014a,b,2015)  
and Webb and Mace (2015)  obtained advected invariant conservation laws in MHD. 

 Cotter et al. (2007) derived multi-symplectic equations for fluid type systems 
by using the momentum map associated with the constraint equations and Clebsch variables of the 
system. Cotter et al. (2007) used the Euler-Poincar\'e approach to Hamiltonian systems 
developed by Holm et al. (1998). Multi-symplectic systems admit conservation laws associated with 
the pullback of the differential forms describing the system to the base manifold, and also satisfy 
the symplecticity conservation laws associated with the conservation of phase space 
following the flow (e.g. Hydon (2005), Bridges et al. (2010)). Noether's theorem for multi-symplectic 
systems is described by Hydon (2005) and Bridges et al. (2010). 
Bridges et al. (2005) show how Ertel's theorem
for ideal, incompressible fluids arises in a multi-symplectic form of the ideal fluid equations.
 Webb et al. (2014c,2015) gave a multi-symplectic formulation of 
MHD by using Clebsch variables in an Eulerian variational principle. 
 Webb and Mace (2015) used Noether's second theorem  
 and a non-field aligned fluid relabelling symmetry to derive a 
potential vorticity type conservation law for MHD. Holm at al. (1983) derived 
 Hamiltonian fluid equations using 
Lagrangian and Eulerian Poisson bracket formulations, semi-direct product Lie algebras, 
and non-canonical Poisson brackets (e.g. Morrison and Greene (1980,1982), Morrison (1982)). 
  Mansfield (2010) and Goncalves and Mansfield 
(2012,2014)  extended the work of Fels and Olver (1998,1999) to 
develop an invariant form of Noether's theorem, using moving frames.

 Our basic variational approach uses Lagrange multipliers to impose constraints on the action. These include, the mass continuity equation, and the entropy advection equation. 
A recent account of the use of Clebsch variables to represent rotational flows, is the work of 
Fukugawa and Fujitani (2010). They obtain Clebsch expansions for the fluid velocity 
${\bf u}$ of the form:
\begin{equation}
{\bf u}=\nabla\phi- r\nabla S-\sum_{\alpha=1}^2 \beta_\alpha \nabla A_\alpha, \label{eq:1.1}
\end{equation}
where the Lagrange multipliers $\phi$ and $r$ ensure that the mass continuity equation and 
the entropy advection equation are satisfied. The vorticity of the fluid from (\ref{eq:1.1}) 
is given by:
\begin{equation}
\boldsymbol{\omega}=\nabla\times {\bf u}=-\nabla r\times\nabla S 
-\sum_{\alpha=1}^2 \nabla\beta_\alpha \times\nabla A_\alpha. \label{eq:1.2}
\end{equation}
Equation (\ref{eq:1.2}) shows in general, that the vorticity $\boldsymbol{\omega}\neq 0$
for isentropic flows, in which $S=const.$. Fukugawa and Fujitani (2010) show that the 
$\beta_\alpha\nabla A_\alpha$ terms in the Clebsch expansion (\ref{eq:1.1}) for ${\bf u}$ 
results from  requiring that the endpoints of the variational path, described by the intersection
of the three surfaces $A_i({\bf x},t)=const.$, ($1\leq i\leq 3$) is required to have zero 
variation $\delta A_i=0$ at the endpoints at the initial and final times $t=t_{init}$ 
and $t=t_{final}$. The $A_i$ are functions of the Lagrange labels, ${\bf a}$ and are advected 
with the flow. The Eulerian density $\rho({\bf x},t)=\rho_o({\bf a}) j$ where $j=\partial(A_1,A_2,A_3)/\partial (x,y,z)$ is the Jacobian of the transformation of the labels $A_i$ and the 
Eulerian position coordinates $(x,y,z)$. The sum in (\ref{eq:1.1}) is the Lin constraint 
term, associated with fluid spin in the absence of entropy gradients. 

Yoshida (2009) studied the Clebsch expansion ${\bf u}=\nabla\phi+\alpha\nabla \beta$
and the completeness of the expansion. The expansion of an arbitrary vector field is incomplete
if the field cannot be expanded in the form (\ref{eq:1.1}). He showed that the generalized 
Clebsch expansion:
\begin{equation}
{\bf u}=\nabla\phi+\sum_{j=1}^\nu \alpha_j\nabla \beta^j, \label{eq:1.3}
\end{equation}
is complete in general if $\nu=n-1$, where $n$ is the number of independent variables.
But, if it is necessary 
to control the boundary values of $\phi$, $\alpha_j$ and $\beta^j$ (e.g. in order to 
determine them uniquely), then $\nu=n$. Russo and Smereka (1999) use a Clebsch description of 
incompressible fluid dynamics using gauge transformations
for the potential $\phi$ in the equations. 

Another approach to Clebsch expansions  
(\ref{eq:1.1}) for ${\bf u}$ is to use gauge field theory (e.g. Kambe (2007,2008), see also 
Jackiw (2002) for the application of gauge field theory to fluid dynamics). In the 
latter approach (Kambe (2008)), the $\beta_\alpha\nabla A_\alpha$ terms 
 in the Clebsch expansion (\ref{eq:1.1}) 
are due to the fluid relabelling symmetries corresponding to rotations in Lagrange label space.
In this approach, the Lagrange multipliers used in the variational principle act as 
gauge potentials, and there are in general gauge transformations that leave ${\bf u}$ invariant.
A related issue for Clebsch potentials is their multi-valued nature. For example, 
for the MHD topological soliton (e.g. Kamchatnov 1982; Semenov et al. 2002) the magnetic field 
induction ${\bf B}=\nabla\times{\bf A}$ has a nontrivial topological structure, which is related 
to the Hopf fibration.  
Semenov et al. (2002) derive complicated magnetic field structures, in which the magnetic vector 
potential has the form ${\bf A}=\alpha\nabla\beta+\nabla K$, in which the potential $K$ 
is not a global potential which has jumps and singularities. Thus, one can obtain magnetic field
structures in which the field lies on a Moebius band. Similar non-trivial 
topological structures with non-global magnetic vector potentials and Clebsch 
potentials arise in the description of the magnetic monopole field (Urbantke 2003).
 
Our analysis uses Clebsch potentials in a Lagrangian variational principle. However, 
we do not need the Clebsch potential form for the fluid velocity in our analysis, 
since we stick with the Lagrangian form of the variational principle. It turns 
out that the Lin constraint terms in the variational principle are decoupled 
from the other Euler Lagrange equations in the Lagrangian variational principle. We discuss 
the connection of our Lagrangian variational principle to an equivalent 
Eulerian variational principle in Appendix A. We impose the condition $\partial x^k(a,t)/\partial t=u^k$ in our variational principle, which is reminiscent of the work of Skinner and Rusk (1983a,
1983b) on generalized Hamiltonian dynamics. The details of our variational approach are described
in Section 3 of the paper. 
  
The main aim of this paper is to extend the multi-symplectic, Lagrangian equations
for  fluid dynamics obtained by Bridges et al. (2005) 
to more general equations of state for the compressible gas dynamics case.   
 Like Bridges et al. 
(2005), we study the connection between the pullback and symplecticity 
conservation laws and vorticity and potential vorticity. 
Bridges et al. (2005)  derived potential vorticity conservation laws  
 associated with the fluid relabelling symmetry (e.g. Padhye and Morrison (1996a,b), Padhye (1998)). 

Section~2 gives the Eulerian fluid equations, 
and a physical discussion of   
 the interaction between the flow kinetic energy and internal energy of the gas. 
Section~3 provides  a 
Lagrangian action principle, in which  Lagrange multipliers are used
 to ensure that the entropy is advected with the flow, and to formally
define the fluid velocity as $u^i=\partial x^i({\bf m},t)/\partial t$. 
An external gravitational potential $\Phi({\bf x})$ 
is included in the Lagrangian to take into account external gravitational fields (e.g. 
 in stellar wind theory, $\Phi({\bf x})$ would be the gravitational potential of the star).  
Because of the non-isobaric equation of state, 
 the conservation laws can be nonlocal, 
since they involve the Lagrange multiplier $r$ used to ensure $dS/dt=0$ in the variational principle.
The variable $r$ is essentially a Clebsch potential 
(e.g. Zakharov and Kuznetsov (1997), Morrison (1998)), which is a nonlocal potential 
not usually regarded as a part of the fluid equations. 
The standard infinite dimensional Hamiltonian functional formulation  and 
Poisson bracket (e.g. Morrison (1998)) is discussed. Section~4 
introduces the de Donder Weyl  multi-momentum formulation.  The multi-symplectic, Lagrangian 
equations, and the pullback and symplecticity conservation laws are obtained in Section~5.  
The symplecticity conservation laws are used to derive Ertel's theorem.
The vorticity-symplecticity conservation laws are also obtained 
 by using a Lie dragging approach (e.g. Tur and Yanovsky (1993),Webb et al. (2014a)). 
The vorticity
flux component that is independent of entropy gradients is Lie dragged by the flow, and gives rise 
to a conservation law analogous to Faraday's law for the advection of magnetic flux in MHD. 
The vorticity-symplecticity  conservation laws  
 are also derived by using Noether's second theorem, in conjunction with  a fluid relabelling 
symmetry due  to mass conservation and by using a divergence symmetry of the 
action (gauge symmetry). 
Section~7 discusses variational principles and the 
Cartan-Poincar\'e form equations for the multi-symplectic equations of Section~5.
A class of exterior differential forms representing the equation system  
 is obtained
(see e.g. Harrison and Estabrook (1971)).  
It is  shown that the ideal of forms 
extracted from the variational principles is a closed ideal of forms that represent the 
multi-symplectic system. 

In Appendix A, we discuss the Lagrangian variational principle
on which our analysis is based. We discuss both the algebraic form of the mass continuity 
equation $J=\tau$ where $J=\det(x_{ij})$ is the determinant of the Lagrangian map, 
and $\tau=1/\rho$ is the specific volume, as well as the derivative form of the mass 
continuity equation $d/dt(J-\tau)=0$.  
 Appendix B  discusses the form of the multi-symplectic equations 
for the case of $n=2$ independent
Lagrangian mass coordinates. Appendix C gives some formulas from Webb et al. (2014c) 
used in Noether's theorem, and Appendix D discusses an application of the Eulerian 
Clebsch variational principle (Zakharov and Kuznetsov (1997)) to verify a conservation law.

Section~8 concludes with a summary 
and discussion. We note that the vorticity-symplecticity conservation law for a non-barotropic 
gas is non-local as it involves the nonlocal Clebsch variable $r$. This implies that the 
 pullback conservation law associated with ${\bf m}$-translation invariance, 
and the vorticity-symplecticity conservation laws are nonlocal. These results are 
clearly of interest in atmospheric dynamics for vortex fluid motions (e.g. tornadoes 
and Rossby waves) where baroclinicity (non-alignment of the gas pressure and density gradients)
will generate vorticity.   
  
\section{Fluid dynamics model}
The time dependent, ideal, inviscid equations of Eulerian gas dynamics, consist of the mass continuity
equation, the Euler momentum equation for the gas, and an equation of state for the gas. 
The mass continuity equation is: 
\begin{equation}
\deriv{\rho}{t}+\nabla{\bf\cdot}(\rho {\bf u})=0. \label{eq:2.1}
\end{equation}
The Euler momentum equation for the fluid can be written in the form:
\begin{equation}
\left(\deriv{\bf u}{t}+{\bf u}{\bf\cdot}\nabla {\bf u}\right)=-\frac{1}{\rho}\nabla p-\nabla\Phi({\bf x}), 
\label{eq:2.2}
\end{equation}
where $\Phi({\bf x})$ is the gravitational potential of an external gravitational field. 
The entropy $S$ is advected with the flow, i.e. 
\begin{equation}
\left(\derv{t}+{\bf u}{\bf\cdot}\nabla\right) S=0. \label{eq:2.3}
\end{equation}
Here $p$, $\rho$, ${\bf u}$ and $S$ are the pressure, density, fluid velocity and entropy 
of the gas respectively.

Equations ({\ref{eq:2.1})-(\ref{eq:2.3}) are supplemented by an equation of state for the gas 
(e.g. $p=p(\rho,S)$), which is related to first law of thermodynamics
by the equation:
\begin{equation}
TdS=dQ=de+p d\tau\quad\hbox{where}\quad \tau=\frac{1}{\rho}, \label{eq:2.4}
\end{equation}
is the specific volume for the gas. For ideal gases $TdS/dt=dQ/dt=0$ where 
$d/dt=\partial/\partial t+{\bf u}{\bf\cdot}\nabla$ is the Lagrangian time derivative 
moving with the flow. The internal energy per unit mass $e$ is related 
to the internal energy per unit volume $\varepsilon(\rho,S)$ by the equation 
$e(\tau,S)=\tau\varepsilon(\rho,S)$. 
Equation  (\ref{eq:2.4}) written in the form:
\begin{equation}
TdS=\frac{1}{\rho}\left(d\varepsilon-wd\rho\right)\quad\hbox{where}\quad w=\frac{\varepsilon+p}{\rho}, 
\label{eq:2.5}
\end{equation}
defines the gas enthalpy $w$. For $\varepsilon=\varepsilon(\rho,S)$, (\ref{eq:2.5}) gives:
\begin{equation}
\rho T=\varepsilon_S,\quad w=\varepsilon_\rho, \quad p=\rho\varepsilon_\rho-\varepsilon. \label{eq:2.6}
\end{equation}
From (\ref{eq:2.4}) we also obtain:
\begin{equation}
TdS=dw-\tau dp\quad\hbox{or}\quad -\frac{1}{\rho}\nabla p=T\nabla S-\nabla w. 
\label{eq:2.7}
\end{equation}

From (\ref{eq:2.5}), the entropy advection equation $TdS/dt=0$ can be written in the form:
\begin{equation}
\frac{d\varepsilon}{dt}-w\frac{d\rho}{dt}=0. \label{eq:2.8}
\end{equation}
Using the mass continuity equation $(1/\rho)d\rho/dt=-\nabla{\bf\cdot}{\bf u}$ 
in (\ref{eq:2.8}), the comoving energy equation (\ref{eq:2.8}) reduces to its 
Eulerian form:
\begin{equation}
\deriv{\varepsilon}{t}+\nabla{\bf\cdot}(\rho {\bf u}w)-{\bf u}{\bf\cdot}\nabla p=0. \label{eq:2.9}
\end{equation}
Taking the scalar product of the Euler momentum equation with $\rho {\bf u}$ 
and using the mass continuity equation (\ref{eq:2.1}) we obtain the kinetic energy and gravitational 
energy equation:
\begin{equation}
\derv{t}\left(\frac{1}{2}\rho u^2+\rho\Phi({\bf x})\right) 
+\nabla{\bf\cdot}\left[\rho {\bf u}\left(\frac{1}{2} u^2
+\Phi({\bf x})\right)\right]+{\bf u}{\bf\cdot}\nabla p=0. 
\label{eq:2.10}
\end{equation}
Adding (\ref{eq:2.9}) and (\ref{eq:2.10}) gives the total energy equation for the system in the form:
\begin{equation}
\derv{t}\left(\frac{1}{2}\rho u^2+\varepsilon(\rho,S)+\rho\Phi({\bf x})\right)
+\nabla{\bf\cdot}\left[\rho {\bf u}\left(\frac{1}{2} u^2+\Phi({\bf x})+w\right)\right]=0. \label{eq:2.11}
\end{equation}
Although (\ref{eq:2.11}) is expected,
both on physical grounds and also on the basis of Noether's first theorem, 
the present discussion emphasizes the intricate coupling between the internal energy 
equation (\ref{eq:2.9}) and the kinetic and gravitational energy equation via the pressure work 
terms $\pm {\bf u}{\bf\cdot}\nabla p$ in (\ref{eq:2.9})-(\ref{eq:2.10}). 

In the next section we describe the Lagrangian action principle approach to the gas dynamic equations. 

\section{Lagrangian gas dynamics}
The gas dynamic equations (\ref{eq:2.1})-(\ref{eq:2.11}) can be derived by requiring 
that the action:
\begin{equation}
{\cal A}=\int{\cal L}\ d^3x dt=\int{\cal L}_0\ d^3x_0dt\equiv \int{\cal L}_m\ d^3mdt, 
\label{eq:3.1}
\end{equation}
is stationary. In (\ref{eq:3.1}) the Lagrangian map is used in which the  
differential equations: $d{\bf x}/dt={\bf u}({\bf x},t)$ are formally integrated for a given 
fluid velocity 
${\bf u}({\bf x},t)$ to obtain the Lagrangian map equations: ${\bf x}={\bf x}({\bf x}_0,t)$ where 
${\bf x}={\bf x}_0$ at time $t=0$. The  map is assumed to be a diffeomorphism, 
i.e. it is 1-1 and invertible, with inverse 
 ${\bf x}_0={\bf x}_0({\bf x},t)$, in which ${\bf x}_0$ is advected with the background flow, 
i.e. 
\begin{equation}
\deriv{{\bf x}_0}{t}+{\bf u}{\bf\cdot}\deriv{{\bf x}_0}{\bf x}=0. \label{eq:3.2}
\end{equation}
We use Lagrange labels ${\bf m}={\bf m}({\bf x}_0)$ so that the 
mass continuity equation 
can be written in the form:
\begin{equation}
\rho d^3x=\rho_0({\bf x}_0) d^3 x_0=d^3 m, \label{eq:3.3}
\end{equation}
which implies
\begin{equation}
\rho J_0=\rho_0\quad\hbox{and}\quad \rho J=1, \label{eq:3.4}
\end{equation}
where
\begin{equation}
J_0=\det\left(\partial x^i/\partial x_0^j\right)\quad \hbox{and}
\quad J=\det\left(\partial x^i/\partial m^j\right). 
\label{eq:3.5}
\end{equation} 
The labels ${\bf m}$ are  Lagrangian mass coordinates, 
and $J_0$ and $J$ are the Jacobians of the Lagrangian maps: ${\bf x}={\bf x}({\bf x}_0,t)$ 
and ${\bf x}={\bf x}({\bf m},t)$ respectively.  The Lagrange labels ${\bf m}={\bf m}({\bf x}_0)$ 
are advected with the background fluid flow (i.e. they satisfy (\ref{eq:3.2}) but with 
${\bf x}_0\to {\bf m}$).  (\ref{eq:3.4}) is equivalent to the Lagrangian mass continuity 
equation. 

For the  model (\ref{eq:2.1})-(\ref{eq:2.11}), the Lagrangian ${\cal L}$ is given by:
\begin{equation}
{\cal L}=\frac{1}{2}\rho u^2-\varepsilon(\rho,S)-\rho\Phi({\bf x}). \label{eq:3.6}
\end{equation}
Using (\ref{eq:3.1})-(\ref{eq:3.6}) gives:
\begin{equation}
{\cal L}_m=\frac{\cal L}{\rho}=\frac{1}{2} u^2-e(\tau,S)-\Phi({\bf x})\quad\hbox{where}\quad 
e(\tau,S)=\frac{\varepsilon(\rho,S)}{\rho}, 
\label{eq:3.7}
\end{equation}
for the Lagrange density in Lagrangian mass coordinates, where 
\begin{equation}
{\bf u}=\deriv{{\bf x}({\bf m},t)}{t}\quad\hbox{and}\quad \rho=\frac{1}{J},
 \label{eq:3.8}
\end{equation}
give  ${\bf u}$ and $\rho$ in terms of the Lagrangian map.

If $x^i$ and $m^j$ are Cartesian coordinates, then the equations:
\begin{align}
&x_{ij}y_{jk}=\delta_{ik},\quad x_{ij}=\deriv{x^i}{m^j}, \quad y_{jk}=\deriv{m^j}{x^k},\nonumber\\
&y_{ij}=\frac{A_{ji}}{J},\quad A_{ij}=\hbox{cofac}(x_{ij}), \label{eq:3.9}
\end{align}
describe the derivatives of the map with respect to ${\bf x}$ and ${\bf m}$.
For the case of $n=3$ space dimensions, $A_{ij}=\rm{cofac}(x_{ij})$ is given by: 
\begin{equation}
A_{ij}=\frac{1}{2}\epsilon_{iab}\epsilon_{jpq}x_{ap}x_{bq}, \label{eq:3.10}
\end{equation}
where $\epsilon_{ijk}$ is the Levi-Civita tensor density (e.g Newcomb (1962), Webb et al. (2005)). 
Other formulae 
for $A_{ij}$ apply in other cases (e.g. for 2 space dimensions). 
These relations can be generalized to generalized coordinates $q^i=q^i({\bf m},t)$, 
but in that case the metric $g_{ij}={\bf e}_i{\bf\cdot}{\bf e}_j$ is important in describing 
the system 
where ${\bf e}_i=\partial {\bf q}/\partial x^i$ are holonomic base vectors. In the present analysis 
we use Cartesian coordinates, ($q^0=t$ and $q^i=x^i$ ($1\leq i\leq n$) are Cartesian coordinates). 
A more general formulation, would use generalized coordinates
and possibly the variational bi-complex.  

We use  the contravariant base vectors 
${e}_\alpha=\partial{\bf x}/\partial m^\alpha$ and the
 dual covariant base vectors ${\bf e}^\alpha=\partial m^\alpha/\partial {\bf x}$. 
One can show that 
\begin{equation}
{\bf e}_\alpha\times {\bf e}_\beta=J \epsilon_{\alpha\beta\gamma}  {\bf e}^\gamma, 
\quad {\bf e}^\alpha\times{\bf e}^\beta=j\epsilon_{\alpha\beta\gamma} {\bf e}_\gamma, 
\quad 1\leq \alpha,\beta\leq n, \label{eq:3.11}
\end{equation}
where $j=\det(y_{ij})=1/J$. 

For the case $n=2$ (i.e. for 2 Cartesian space dimensions), the co-factor matrix $A_{ij}$ is given by:
\begin{equation}
{\sf A}_{ij}=\left(\begin{array}{cc}
x_{22}& -x_{21}\\
-x_{12}& x_{11}
\end{array}
\right). \label{eq:3.12}
\end{equation}

The Lagrangian mass continuity equation has two aspects. From (\ref{eq:3.3})-(\ref{eq:3.4}) we require
\begin{equation}
\rho J=1\quad\hbox{or}\quad J-\tau=0,  \label{eq:3.13}
\end{equation}
We also require that 
\begin{equation}
\frac{d}{dt}(\tau-J)=\tau_t-\frac{dJ}{dt}=0. \label{eq:3.14}
\end{equation}
The latter equation, can be written in the form:
\begin{equation}
\tau_t-\deriv{J}{x_{ij}}\deriv{x_{ij}}{t}=\tau_t-A_{ij}\deriv{u^i}{m^j}=0, \label{eq:3.15}
\end{equation}
(note that $\partial x_{ij}/\partial t=\partial u^i/\partial m^j$). 
 Taking into account the constraints (\ref{eq:3.13})-(\ref{eq:3.15}) we introduce 
the constrained Lagrangian:
\begin{align}
\ell_m=&\frac{1}{2} u^2-e(\tau,S) -\Phi({\bf x}) +r \frac{dS}{dt}
+\lambda^k \frac{d\mu^k}{dt}\nonumber\\
&+\Lambda^i\left(u^i-\deriv{x^i}{t}\right)+\nu(J-\tau)
+\zeta\frac{d}{dt}\left(\tau-J\right),
 \label{eq:3.16}
\end{align}
where the $\mu^k$ correspond to the so-called Lin constraints (e.g.  
 Holm and Kupershmidt (1983a,b)). 
The continuity constraint terms in (\ref{eq:3.16}) can be written in the form:
\begin{equation}
(\nu+\zeta_t)(J-\tau)+\frac{d}{dt}\left[\zeta(\tau-J)\right], \label{eq:3.16aa}
\end{equation}
Using (\ref{eq:3.16aa}) in (\ref{eq:3.16}) it follows that the Lagrangian $\ell_m$ in 
(\ref{eq:3.16}) can be replaced by 
\begin{align}
\ell_m=&\frac{1}{2} u^2-e(\tau,S) -\Phi({\bf x}) +r \frac{dS}{dt}
+\lambda^k \frac{d\mu^k}{dt}\nonumber\\
&+\Lambda^i\left(u^i-\deriv{x^i}{t}\right)
+\left(\nu+\zeta_t\right) (J-\tau),
 \label{eq:3.16ab}
\end{align}
because two Lagrangians $L_1$ and $L_2$ which differ by 
a perfect derivative or divergence term have the same Euler Lagrange equations 
(e.g. Bluman and Kumei (1989)). 
The continuity equation constraint term in (\ref{eq:3.16ab}) can be written in the forms:
\begin{equation}
(\nu+\zeta_t)(J-\tau)\equiv\tilde{\nu}(J-\tau)\equiv\tilde{\zeta}_t(J-\tau)\quad\hbox{where}\quad \tilde{\nu}=\tilde{\zeta}_t=\nu+\zeta_t.  \label{eq:3.16ac}
\end{equation}
Thus, the continuity equation constraint  involves only one Lagrange multiplier, i.e. 
one can use either $\tilde{\nu}$ or $\tilde{\zeta}_t$ as the Lagrange multiplier. Since
\begin{equation}
\tilde{\zeta}_t\left(J-\tau\right)\equiv 
\frac{d}{dt}\left(\tilde{\zeta}(J-\tau)\right)
-\tilde{\zeta} \frac{d}{dt}(J-\tau), \label{eq:3.16ad}
\end{equation}
then 
\begin{align}
\ell_m=&\frac{1}{2} u^2-e(\tau,S) -\Phi({\bf x}) +r \frac{dS}{dt}
+\lambda^k \frac{d\mu^k}{dt}\nonumber\\
&+\Lambda^i\left(u^i-\deriv{x^i}{t}\right)
+\tilde{\zeta}\frac{d}{dt}(\tau-J), 
 \label{eq:3.16ae}
\end{align}
is another form of $\ell_m$ that gives the correct Euler Lagrange equations. 
This latter form of $\ell_m$ is useful 
in the Eulerian Clebsch potential formulation of the variational principle as developed 
by Zakharov and Kuznetsov (1997) (see Appendix A).

Below we use the Lagrangian:
\begin{equation}
\ell_{m0}=\frac{1}{2} u^2-e(\tau, S)-\Phi({\bf x})+ r\frac{dS}{dt}+\Lambda^i\left(u^i-\deriv{x^i}{t}\right)
+\tilde{\nu}(J-\tau) +\lambda^k \frac{d\mu^k}{dt}. \label{eq:3.16a}
\end{equation}
 It turns out that the Euler-Lagrange equations using the Lagrangian (\ref{eq:3.16a})
for the $\mu^k$ and $\lambda^k$ Clebsch 
potentials decouple from the other Euler-Lagrange equations. 
We show that the $\mu^k$ and $\lambda^k$ do not contribute to the pullback and symplecticity 
conservation laws in Section 5. This is a consequence of the fact that the $\mu^k$ and $\lambda^k$ 
are advected with the fluid.

The constraint terms in the Lagrangian (\ref{eq:3.16a}) are nonholonomic constraints, and are 
 examples of the Skinner-Rusk construction [Skinner and Rusk (1983a, 1983b), 
Cantrijn and Vankerschaver (2007), Llibre et al. (2014)]. 
 The nonholonomic constraint terms 
are sometimes referred to as vakonomic constraints ( e.g. Llibre et al. (2014); 
the term vakonomic stands for ``variational axiomatic kind" a term coined by Kozlov). 
In the term $\Lambda^i(u^i-\partial x^i/\partial t)$ in (\ref{eq:3.16a}) 
${\bf u}-\partial {\bf x}/\partial t$ lies in the tangent space $TQ$ 
 and $\Lambda ^i$ is the $ith$ 
component of a co-vector in the co-tangent space $T^*Q$ (i.e. the dual of the vector space 
$TQ$). The Skinner-Rusk construction does not necessarily imply a  standard Hamiltonian system, 
in which the canonical momenta are constructed by the Legendre transformation, 
since the Lagrangian may be singular.

\subsection{The Euler Lagrange equations}
In this section, we obtain the Euler Lagrange equations for the action (\ref{eq:3.1}) 
in which the Lagrangian $\ell_m$  is given by (\ref{eq:3.16a}) and 
has the functional form:
\begin{equation}
\ell_m=\ell_m({\bf x},{\bf u}, \tau,S,r,\boldsymbol{\Lambda}, x_{it},\nu,x_{ij},\mu^k,\lambda^k), 
\label{eq:el1}
\end{equation}
where $x_{ij}=\partial x^i/\partial m^j$ and $x_{it}=\partial x^i/\partial t$. 
The stationary point conditions for the action (\ref{eq:3.1}) with ${\cal{L}}_m\to\ell_m$ 
where $\ell_m$ is given by (\ref{eq:3.16a}), give the constraint equations: 
\begin{align}
\frac{\delta {\cal A}}{\delta\tau}=&- e_\tau-\tilde{\nu}=0\quad \hbox{or}
\quad \tilde{\nu}=-e_\tau=p, 
\label{eq:el2}\\
\frac{\delta{\cal A}}{\delta\tilde{\nu}}=&J-\tau=0, \label{eq:el3}\\
\frac{\delta {\cal A}}{\delta\Lambda^i}=&u^i-\deriv{x^i}{t}=0, \label{eq:el4}\\
\frac{\delta {\cal A}}{\delta u^i}=&\deriv{\ell_m}{u^i}
=u^i+\Lambda^i=0\quad\hbox{or}\quad \Lambda^i=-u^i, 
\label{eq:el5}\\
\frac{\delta {\cal A}}{\delta r}=&\deriv{\ell_m}{r}=\frac{dS}{dt}=0, \label{eq:el6}\\
\frac{\delta {\cal A}}{\delta S}=&\deriv{\ell_m}{S}-\derv{t}\left(\deriv{\ell_m}{S_t}\right)
=-e_S-r_t=-\left(r_t+T\right)=0, \label{eq:el7}\\
 \frac{\delta {\cal A}}{\delta \lambda^k}=&\frac{d\mu^k}{dt}=0, \quad 
\frac{\delta {\cal A}}{\delta \mu^k}=-\frac{d\lambda^k}{dt}=0. \label{eq:el7a}
\end{align}
where $T$ is the temperature of the gas.

The stationary point conditions for ${\cal A}$ due to variations of the $x^i$ give the 
Euler-Lagrange equations:
\begin{align}
\frac{\delta {\cal A}}{\delta x^i}&=
\deriv{\ell_m}{x^i}-\derv{t}\left(\deriv{\ell_m}{x_{it}}\right)
-\derv{m^j}\left(\deriv{\ell_m}{x_{ij}}\right)\nonumber\\
&=-\deriv{\Phi}{x^i}-\deriv{(-\Lambda^i)}{t}
-\derv{m^j}\left(\tilde{\nu} \deriv{J}{x_{ij}}\right)\nonumber\\
&\equiv -\biggl(\frac{\partial u^i({\bf m},t)}{\partial t}+\deriv{\Phi}{x^i}
+\derv{m^j}\left(pA_{ij}\right)\biggr)=0, \label{eq:el8}
\end{align}
(note  $\Lambda^i=-u^i$ from (\ref{eq:el5}) and $\tilde{\nu}=p$ from (\ref{eq:el2})). 
Equation (\ref{eq:el8}) is the Lagrangian momentum equation.  By noting that 
$\partial A_{ij}/\partial m^j=0$ (e.g. Newcomb (1962)), 
(\ref{eq:el8}) can be re-written as:
\begin{equation}
\frac{\delta {\cal A}}{\delta x^i}=-\left(\frac{du^i}{dt}
+\deriv{\Phi}{x^i} +\frac{1}{\rho}\deriv{p}{x^i}\right)=0, \label{eq:el9}
\end{equation}
which is equivalent to the Eulerian momentum equation (\ref{eq:2.2}). 
 Notice that the Euler-Lagrange (EL)
equations (\ref{eq:el7a}) are independent of the preceeding Euler Lagrange  
equations (\ref{eq:el2})-(\ref{eq:el7}) and of (\ref{eq:el9}). 
It is important to retain
the $\mu^k$ and $\lambda^k$ Clebsch potentials in the Poisson bracket. For example, if 
one wishes to transform the Poisson bracket to noncanonical coordinates, the 
$\mu^k$ and $\lambda^k$ are Casimirs: (e.g. Morrison and Greene (1980,1982), Morrison (1998)). 

Substituting $\tilde{\nu}$ and the $\Lambda^i$ from (\ref{eq:el2}) 
and (\ref{eq:el5}), in (\ref{eq:3.9}) gives:
\begin{align}
\tilde{\ell}_m=&\frac{1}{2}u^2-e(\tau,S)-\Phi({\bf x})+r \deriv{S({\bf m},t)}{t}
+u^i\left(\deriv{x^i}{t}-u^i\right)+p(J-\tau)+\lambda^k \frac{d\mu^k}{dt}\nonumber\\
=&r\deriv{S}{t}+u^i \deriv{x^i}{t}+p\det\left(x_{ij}\right)+\lambda^k \frac{d\mu^k}{dt}
-\left(\frac{1}{2}u^2+\tilde{w}(p,S)+\Phi({\bf x})\right), \label{eq:el10}
\end{align}
as an equivalent form of $\ell_m$ where $\tilde{w}(p,S)=w(\rho,S)$ is the enthalpy
of the gas. Note that the Lagrangian $\tilde{\ell}_m$ in (\ref{eq:el10}) has the functional form:
\begin{equation}
\tilde{\ell}_m=\tilde{\ell}_m\left({\bf x},{\bf u},S,p,r,S_t,x_{it},x_{ij},
\mu^k_t,\lambda^k\right), \label{eq:el11}
\end{equation}
and that
\begin{equation}
\tilde{w}_p=\tau \quad\hbox{and}\quad \tilde{w}_S=T\label{eq:el12}
\end{equation}
(see Webb (2015)). 

Using $\tilde{\ell}_m$ from (\ref{eq:el10}) to replace ${\cal L}_m$ in the action (\ref{eq:3.1}), 
we obtain the stationary point conditions for the action as:
\begin{align}
\frac{\delta {\cal A}}{\delta p}=&J-\tilde{w}_p=J-\tau=0, \label{eq:el13}\\
\frac{\delta {\cal A}}{\delta u^i}=&\deriv{\tilde{\ell}_m}{u^i}=\deriv{x^i}{t}-u^i=0, 
\label{eq:el14}\\
\frac{\delta {\cal A}}{\delta r}=&\deriv{\tilde{\ell}_m}{r}=\deriv{S}{t}=0, \label{eq:el15}\\
\frac{\delta {\cal A}}{\delta S}=&\deriv{\tilde{\ell}_m}{S}
-\derv{t}\left(\deriv{\tilde{\ell}_m}{S_t}\right)
=-\tilde{w}_S-\deriv{r}{t}=-\left(\frac{dr}{dt}+T\right)=0, \label{eq:el16}\\
\frac{\delta {\cal A}}{\delta x^i}=&\deriv{\tilde{\ell}_m}{x^i}-\derv{t}\left(\deriv{\tilde{\ell}_m}{x_{it}}\right)
-\derv{m^j}\left(\deriv{\tilde{\ell}_m}{x_{ij}}\right)\nonumber\\
=&-\deriv{\Phi}{x^i}-\deriv{u^i}{t}-\derv{m^j}\left(p A_{ij}\right)\nonumber\\
=&-\left(\frac{du^i}{dt}+\deriv{\Phi}{x^i}+\derv{m^j}\left(p A_{ij}\right)\right)=0, 
\label{eq:el17}\\
\frac{\delta {\cal A}}{\delta\lambda^k}=&\frac{d\mu^k}{dt}=0,\quad 
\frac{\delta {\cal A}}{\delta \mu^k}=-\frac{d\lambda^k}{dt}=0.
\label{eq:el17a}
\end{align}
Thus action (\ref{eq:3.1}) with action density $\tilde{\ell}_m$ from (\ref{eq:el10}) gives the 
 equations (\ref{eq:el2})-(\ref{eq:el8}) obtained previously using the 
Lagrange multipliers $\tilde{\nu}$ and $\Lambda^i$.

\subsection{Standard Hamiltonian approach}
In the standard Hamiltonian approach, in which the evolution variable is time $t$, 
one defines the canonical momenta  by the equations:
 \begin{equation}
\pi_k=\deriv{\tilde{\ell}_m}{x^k_t}=u^k, \quad \pi_S=\deriv{\tilde{\ell}_m}{S_t}=r, 
\quad \pi_{\mu^k}=\deriv{\tilde{\ell}_m}{\mu^k_t}=\lambda^k. \label{eq:ell18}
\end{equation}
Use of the Legendre transform gives:
\begin{equation}
h_c=\pi_k\deriv{x^k}{t}+\pi_S S_t+ \pi_{\mu^k} \frac{d\mu^k}{dt}-\tilde{\ell}_m
\equiv\frac{1}{2}u^2+e(\tau,S)+\Phi({\bf x}), \label{eq:ell19}
\end{equation}
for the classical Hamiltonian density $h_c$, where $\tau\equiv J$.  
The Hamiltonian functional $H_c$ is defined as:
\begin{equation}
H_c=\int h_c\ d^3m\equiv 
\int\left(\frac{1}{2}u^2+e(\tau,S)+\Phi({\bf x})\right)\ d^3 m.
\label{eq:ell20}
\end{equation}

Taking the variational derivative of $H_c$ with respect to $x^k$ gives:
\begin{align}
\frac{\delta H_c}{\delta x^k}&=\deriv{h_c}{x^k}-\derv{m^j}
\left(\deriv{h_c}{x_{kj}}\right)\nonumber\\
&=\deriv{\Phi}{x^k}+\derv{m^j}\left(p A_{kj}\right)
=\deriv{\Phi}{x^k}+\frac{1}{\rho}\deriv{p}{x^k}. 
\label{eq:ell21}
\end{align}
Thus, 
\begin{equation}
\frac{\delta H_c}{\delta x^k}=\left(\deriv{\Phi}{x^k}+\frac{1}{\rho}\deriv{p}{x^k}\right)
=-\frac{du^k}{dt}. \label{eq:ell22}
\end{equation}
In deriving (\ref{eq:ell22}) we used the facts: $J=\tau$, $\partial J/\partial x_{kj}=A_{kj}$ 
and $e_\tau=-p$. 
 Similarly,
\begin{align}
\frac{\delta H_c}{\delta u^k}&=u^k=\frac{dx^k}{dt}, \label{eq:ell23}\\
\frac{\delta H_c}{\delta S}&=T=-\frac{dr}{dt},\quad \frac{\delta H_c}{\delta r}=0=\frac{dS}{dt}, 
\label{eq:ell24}\\
\frac{d\mu^k}{dt}=&\frac{\delta H_c}{\delta\lambda^k}=0,\quad 
\frac{d\lambda^k}{dt}=-\frac{\delta H_c}{\delta \mu^k}=0. \label{eq:ell24a}
\end{align}}
Equations (\ref{eq:ell22})-(\ref{eq:ell24a}) are Hamilton's canonical equations for Lagrangian 
gas dynamics. i.e. 
 \begin{align}
\frac{dx^k}{dt}&=\frac{\delta H_c}{\delta u^k},
\quad \frac{du^k}{dt}=-\frac{\delta H_c}{\delta x^k}, \nonumber\\ 
\frac{dS}{dt}&=\frac{\delta H_c}{\delta r}=0,\quad \frac{dr}{dt}
=-\frac{\delta H_c}{\delta S}=-T, \nonumber\\
\frac{d\mu^k}{dt}=&\frac{\delta H_c}{\delta\lambda^k}=0,\quad 
\frac{d\lambda^k}{dt}=-\frac{\delta H_c}{\delta \mu^k}=0. \label{eq:ell25} 
\end{align}

The canonical Poisson bracket for the Hamiltonian system (\ref{eq:ell25}) is:
\begin{equation}
\left\{F,H\right\}=\int \left(\frac{\delta F}{\delta x^k}\frac{\delta H}{\delta u^k} 
-\frac{\delta F}{\delta u^k}\frac{\delta H}{\delta x^k} 
+\frac{\delta F}{\delta S}\frac{\delta H}{\delta r} 
-\frac{\delta F}{\delta r}\frac{\delta H}{\delta S}
+\frac{\delta F}{\delta \mu^k}\frac{\delta H}{\delta \lambda^k}
-\frac{\delta F}{\delta \lambda^k}\frac{\delta H}{\delta \mu^k}\right)\ d^3m. \label{eq:ell26}
\end{equation}
Using the Poisson bracket (\ref{eq:ell26}), Hamilton's equations for functionals $F$ 
of the canonical variables can be written in the form $\dot{F}=\left\{F,H_c\right\}$ 
which gives the time evolution of the functional $F$. Noncanonical forms of the Poisson 
bracket in terms of physical variables may  be obtained by transforming the 
variational derivatives in the Poisson bracket (\ref{eq:ell26}) to the new, noncanonical 
variables 
(e.g. Morrison and Greene (1980,1982); Holm and Kupershmidt (1983a,b); 
Holm, Kuperschmidt and Levermore (1983), Webb et al. (2014a)).

\section{de Donder Weyl multi-momentum approach}
The formal mathematical and theoretical physics approach to multi-symplectic 
systems uses the language of fiber bundles and jet bundles. In this approach physical 
fields are thought of as sections of vector bundles (sectioning means the 
imposition of the dependence of the physical variables on the independent variables 
in the system). Gotay (2004a,b) gives a physics oriented description of this approach, 
and uses the simple example of particle dynamics in Hamiltonian mechanics,  
 which can be generalized to more complex systems. 
 We investigate the effect of the Clebsch potential terms $\lambda^k d\mu^k/dt$ 
terms in the Lagrangian (\ref{eq:el10}).

To derive the de-Donder Weyl equations, 
we introduce canonical multi-momenta associated with the 
Lagrange density $\tilde{\ell}_m$ in (\ref{eq:el10}), namely,
\begin{equation}
\pi_{kt}=\deriv{\tilde{\ell}_m}{x^k_t}=u^k, \quad \pi_{kj}=\deriv{\tilde{\ell}_m}{x_{kj}}=pA_{kj},
\quad \pi_{St}=\deriv{\tilde{\ell}_m}{S_t}=r,\quad  
\pi_{\mu^k_t}=\lambda^k. \label{eq:4.1}
\end{equation}
In the de Donder-Weyl approach, both the time $t$ and the Lagrange labels $m^j$ can be thought 
of as evolution variables. Using (\ref{eq:el10}) for $\tilde{\ell}_m$ and the 
multi-momenta (\ref{eq:4.1}),  the generalized Legendre transformation:
\begin{equation}
h=\pi_{kt}\deriv{x^k}{t}+\pi_{kj} x_{kj}+\pi_{St} S_t+\pi_{\mu^k_t}\frac{d\mu^k}{dt}
-\tilde{\ell}_m, \label{eq:4.2}
\end{equation}
gives the multi-symplectic Hamiltonian density $h$.

To derive the de Donder-Weyl equations, we note that:  
\begin{align}
h&=h(x^k,S,\pi_{kt},\pi_{kj},\pi_{St}, \pi_{\mu^k_t}), \nonumber\\
\tilde{\ell}_m&=\tilde{\ell}_m(x^k,u^k,S,r,p,x^k_t,x_{kj},S_t,\lambda^k,\mu^k_t). \label{eq:4.5}
\end{align}
Using (\ref{eq:4.5})
and taking the differential of (\ref{eq:4.2}) gives:
\begin{align}
dh&=\deriv{h}{x^k} dx^k+\deriv{h}{S} dS+\frac{\partial h}{\partial \pi_{kt}} d\pi_{kt}
+\frac{\partial h}{\partial \pi_{St}}dS_t+\frac{\partial h}{\partial \pi_{kj}} d\pi_{kj}+
\frac{\partial h}{\partial\pi_{\mu^k_t}}
d\pi_{\mu^k_t}\nonumber\\
&=d\left(\pi_{St}S_t+\pi_{kt}x^k_t+\pi_{kj} x_{kj}+\pi_{\mu^k_t}\mu^k_t\right)\nonumber\\
&\ \ -\biggl(\deriv{\tilde{\ell}_m}{x^k} dx^k+ \deriv{\tilde{\ell}_m}{x^k_t} dx^k_t 
+\deriv{\tilde{\ell}_m}{S} dS +\deriv{\tilde{\ell_m}}{x_{kj}} dx_{kj}
+\deriv{\tilde{\ell}_m}{S_t} dS_t+\deriv{\tilde{\ell}_m}{r} dr
+\deriv{\tilde{\ell}_m}{u^k} du^k\nonumber\\
&\ \ +\deriv{\tilde{\ell}_m}{\lambda^k} d\lambda^k
+\frac{\partial\tilde{\ell}_m}{\partial\mu^k_t} d\mu^k_t\biggr). \label{eq:4.6}
\end{align}
Using the Euler Lagrange equations (\ref{eq:el6})-(\ref{eq:el7}) 
gives:
\begin{align}
\frac{\delta {\cal A}}{\delta r}&=\deriv{\tilde{\ell}_m}{r}=S_t=0,\label{eq:4.7}\\
\frac{\delta {\cal A}}{\delta S}&=\deriv{\tilde{\ell}_m}{S}
-\derv{t}\left(\deriv{\tilde{\ell}_m}{S_t}\right)=-\frac{\varepsilon_S}{\rho}-r_t
\equiv -\left(r_t+T\right)=0. \label{eq:4.8}
\end{align}
Equating the $dS$ terms in (\ref{eq:4.6}) gives:
\begin{equation}
\deriv{h}{S}=-\deriv{\tilde{\ell}_m}{S}=T=-r_t. \label{eq:4.9}
\end{equation}

Equating the various differentials in (\ref{eq:4.6}) gives rise to the 
following equations:
\begin{align}
&du^k:\quad\quad \frac{\delta {\cal A}}{\delta u^k}=\deriv{\tilde{\ell}_m}{u^k}
=\left(\deriv{x^k}{t}-u^k\right)=0, \label{eq:4.10}\\
&dS_t:\quad\quad -\deriv{\tilde{\ell}_m}{S_t}+\pi_{St}=0\quad\hbox{or}\quad\pi_{St}=r, 
\label{eq:4.11}\\
&d\pi_{St}:\quad\quad \deriv{h}{\pi_{St}}=S_t=\deriv{\tilde{\ell}_m}{r}
\quad\hbox{or}\quad S_t=\deriv{h}{\pi_{St}}=0, \label{eq:4.12}\\
&d\pi_{kt}:\quad\quad \deriv{x^k}{t}=\deriv{h}{\pi_{kt}}, \label{eq:4.13}\\
&dx^k_t:\quad\quad \pi_{kt}=\deriv{\tilde{\ell}_m}{x^k_t}=u^k, \label{eq:4.14}\\
&d\pi_{kj}:\quad\quad \deriv{x^k}{m^j}=\deriv{h}{\pi_{kj}}, \label{eq:4.15}\\
&dx_{kj}:\quad\quad \pi_{kj}=\deriv{\tilde{\ell}_m}{x_{kj}}=p A_{kj},
\quad \deriv{h}{x_{kj}}=\left(\pi_{kj}-p A_{kj}\right)=0, \label{eq:4.16}\\
&dx^k:\quad\quad \deriv{h}{x^k}=-\deriv{\tilde{\ell}_m}{x^k}=\deriv{\Phi}{x^k}. \label{eq:4.17}
\end{align}
The balance equations for $d\lambda^k$ and $d\mu^k_t$ and $d\pi_{\mu^k_t}$ give the 
equations: 
\begin{equation}
\frac{d\mu^k}{dt}=\deriv{h}{\lambda^k}=0, \quad \pi_{\mu^k_t}=\lambda^k, 
\quad \frac{d\lambda^k}{dt}=-\deriv{h}{\mu^k}=0. \label{eq:4.17a}
\end{equation}
Note that (\ref{eq:4.16}) implies:
\begin{equation}
0=\deriv{h}{x_{kj}}=\pi_{kj}-p A_{kj}, \label{eq:4.17b}
\end{equation}
The latter equation implies that there is no evolution of $x_{kj}$ with respect to $(t,m^1,m^2,m^3)$
in the multi-symplectic
Hamiltonian formulation described in the next section.

The Euler-Lagrange equation (\ref{eq:el8}) can be written in the form:
\begin{align}
\frac{\delta {\cal A}}{\delta x^k}&=\deriv{\tilde{\ell}_m}{x^k}
-\derv{t}\left(\deriv{\tilde{\ell}_m}{x^k_t}\right)
-\derv{a^j}\left(\deriv{\tilde{\ell}_m}{x_{kj}}\right)\nonumber\\
&\equiv-\deriv{\pi_{kt}}{t}-\deriv{\pi_{kj}}{m^j}+\deriv{\tilde{\ell}_m}{x^k}=0. \label{eq:4.18}
\end{align}
Thus, (\ref{eq:4.18}) gives the Hamiltonian divergence like equation:
\begin{equation}
\deriv{\pi_{kt}}{t}+\deriv{\pi_{kj}}{m^j}
=-\deriv{h}{x^k}=-\frac{\delta H}{\delta x^k}, 
\label{eq:4.19}
\end{equation}
where 
\begin{equation}
H=\int h\ d^3m, \label{eq:4.20}
\end{equation}
is the Hamiltonian functional. Equation (\ref{eq:4.19}) is equivalent to the 
Euler momentum equation (\ref{eq:el7}) or (\ref{eq:el8}).

To sum up, (\ref{eq:4.13}), (\ref{eq:4.15}) and (\ref{eq:4.19}) give the 
de Donder-Weyl Hamiltonian equations:
\begin{equation}
\nabla{\bf\cdot}\Pi_k=-\frac{\delta H}{\delta x^k},\quad 
\deriv{x^k}{t}=\frac{\delta H}{\delta \pi_{kt}}, \quad
\deriv{x^k}{m^j}=\frac{\delta H}{\delta \pi_{kj}}, \label{eq:4.21}
\end{equation}
where
\begin{equation}
\nabla{\bf\cdot}\Pi_k\equiv \deriv{\pi_{kt}}{t}+\deriv{\pi_{kj}}{m^j}. \label{eq:4.22}
\end{equation}
Similarly, (\ref{eq:4.8}),(\ref{eq:4.9}) and (\ref{eq:4.12}) give the Hamiltonian 
equations for $S$ and $r$, and (\ref{eq:4.17a}) 
give equations for $\mu^k$ and $\lambda^k$:
\begin{equation}
S_t=\frac{\delta H}{\delta r}=0,\quad r_t=-\frac{\delta H}{\delta S}=-T,
\quad \frac{d\mu^k}{dt}=\frac{\delta H}{\delta\lambda^k}=0,
\quad \frac{d\lambda^k}{dt}=-\frac{\delta H}{\delta\mu^k}=0. \label{eq:4.23}
\end{equation}

\section{Multi-symplectic formulation}
The Lagrangian gas dynamical system can be written in the multi-symplectic form:
\begin{equation}
{\sf K}^0_{is}\deriv{z^s}{t}+{\sf K}^k_{is}\deriv{z^s}{m^k}=\frac{\delta H}{\delta z^i}, 
\quad 1\leq i\leq N,
\label{eq:5.1}
\end{equation}
where $N$ is the number of variables $z^s$. In (\ref{eq:5.1}) 
\begin{equation}
H=\int h_m d{\bf m}, \label{eq:5.1a}
\end{equation}
is the multi-symplectic Hamiltonian functional and
\begin{equation}
 \quad h_m= \frac{1}{2} \langle{\bf u},{\bf u}\rangle +e(\tau,S)+\Phi({\bf x})+\pi_{ik} x_{ik},
\quad e(\tau,S)=\frac{\varepsilon}{\rho} \label{eq:5.1b}
\end{equation}
defines the multi-symplectic Hamiltonian density in which $e(\tau,S)$ is the internal energy 
density of the gas per unit mass and $\tau=1/\rho\equiv J=\det(x_{ij})$ 
is the specific volume of the gas.  
The dependent variables $z^s$ in (\ref{eq:5.1}) 
are the same variables that appear in the de Donder-Weyl 
 Hamiltonian formulation of Section~4, but include also the independent variables $x_{kj}$. 
The $m^k, 1\leq k\leq n$ are the Lagrangian fluid labels.
$H$  is the de Donder-Weyl Hamiltonian functional 
given by (\ref{eq:5.1a}) and (\ref{eq:5.1b}). In (\ref{eq:5.1}) the matrices ${\sf K}^\alpha_{ij}$ 
are skew symmetric in the two lower indices, and are related to fundamental one-forms 
describing the system, which in our case are related to the Legendre transformation 
used in (\ref{eq:4.2}). 

Below, we develop the equations for the case of $n=3$ independent Lagrangian mass coordinates. 
We indicate in Appendix B, how the formalism also applies for the case $n=2$. The case 
$n=1$ of 1D gas dynamics is described by Webb (2015). The same basic equations and principles 
apply for all values of $n$.  
The dependent variables $z^s$ in (\ref{eq:5.1}) are given by:
\begin{align}
{\bf z}=&\biggl(x^1,x^2,x^3,u^1,u^2,u^3,\left\{\pi_{ij}:\quad 1\leq i\leq 3,\quad 1\leq j\leq 3\right\},
S,r, \nonumber\\
&\left\{ x_{ij}:\quad 1\leq i\leq 3,\quad 1\leq j\leq 3\right\},
\mu^1,\lambda^1,\mu^2,\lambda^2,\ldots
\biggr)^T. \label{eq:5.2}
\end{align}
 Alternatively, we write:
\begin{equation}
{\bf z}=\left({\bf x}^T,\left(\pi_{{\bf x}t}\right)^T,\pi_{ij},S,\pi_{St},
\left\{x_{ij}:\quad 1\leq i\leq 3,\quad 1\leq j\leq 3\right\},
\mu^1,\lambda^1,\mu^2,\lambda^2,\ldots \right)^T. \label{eq:5.3}
\end{equation}
Thus, the variables ${\bf z}$ consist of the coordinates $(x^1,x^2,x^3,S)$ and the multi-momenta 
$(\pi_{it},\pi_{ij},\pi_{St})^T$,  the $x_{ij}$ and the $\mu^k$ and $\lambda^k$. In (\ref{eq:5.3}) the $\pi_{ij}$ are ordered, so that the column index varies first from $1\leq j\leq 3$ and then the $i$ index increments by 1, 
and the cycle is repeated for $i=2$ and $i=3$. Thus, we use the convention:
\begin{align}
&(z^1,z^2,z^3)=(x^1,x^2,x^3),\quad (z^4,z^5,z^6)=(u^1,u^2,u^3), \nonumber\\
&(z^7,z^8,z^9,z^{10},z^{11},z^{12},z^{13},z^{14},z^{15})=(\pi_{11},\pi_{12},\pi_{13},
\pi_{21},\pi_{22},\pi_{23},\pi_{31},\pi_{32},\pi_{33})\nonumber\\
&(z^{16},z^{17})=(S,r), \nonumber\\
&(z^{18},z^{19},z^{20},z^{21},z^{22},z^{23},z^{24},z^{25},z^{26})=(x_{11},x_{12},x_{13},
x_{21},x_{22},x_{23},x_{31},x_{32},x_{33}),\nonumber\\
&(z^{27},\ldots )=(\mu^1,\lambda^1,\mu^2,\lambda^2,\ldots). 
 \label{eq:5.4}
\end{align}

We introduce the one-forms:
\begin{equation}
\omega^\alpha=L^\alpha_s dz^s, \quad 0\leq\alpha\leq 3,\quad 1\leq s\leq N, \label{eq:5.5}
\end{equation}
where
\begin{equation}
L^0_s\deriv{z^s}{t}+L^j_s \deriv{z^s}{m^j}=\pi_{kt}\deriv{x^k}{t}+\pi_{St}\deriv{S}{t}
+\pi_{kj}x_{kj}+\pi_{\mu^k_t}\frac{d\mu^k}{dt}, \label{eq:5.6}
\end{equation}
are the multi-momenta terms in the Legendre transformation (\ref{eq:4.2}). Thus we obtain:
\begin{equation}
\omega^0=\pi_{x^it} dx^i+\pi_{St} dS+\pi_{\mu^k_t} d\mu^k=u^i dx^i+r dS+\lambda^k d\mu^k\equiv L^0_s dz^s, \label{eq:5.7}
\end{equation}
Similarly, we set:
\begin{equation}
\omega^k=\pi_{ik}dx^i=p A_{ik}dx^i\equiv L^k_s dz^s, 
\quad 1\leq i\leq 3,\quad 1\leq k\leq 3, \label{eq:5.8}
\end{equation}
for the one-forms associated with the multi-momenta $\pi_{ik}$. 

 
In the multi-symplectic formulation, the fundamental two-forms $\kappa^\alpha=d\omega^\alpha$ 
are closed forms, i.e. $d\kappa^\alpha=0$ ($0\leq\alpha\leq 3$). Thus 
$\omega^\alpha=L^\alpha_s dz^s$ are such that 
\begin{equation}
\kappa^\alpha=d\omega^\alpha=d\left(L^\alpha_j dz^j\right)
=\frac{1}{2} {\sf K}^\alpha_{ij}dz^i\wedge dz^j. \label{eq:5.9}
\end{equation}
From (\ref{eq:5.9}) the matrices ${\sf K}^\alpha_{ij}$ have the form:
\begin{equation}
{\sf K}^\alpha_{ij}=\deriv{L^\alpha_j}{z^i}-\deriv{L^\alpha_i}{z^j}. \label{eq:5.10}
\end{equation}
The matrices ${\sf K}^\alpha_{ij}$ are skew-symmetric with respect to the two lower indices $i$ and $j$ 
(e.g. Hydon (2005), Cotter et al. (2007), Webb et al. (2014c)). 

The components of the matrices ${\sf K}^0_{ij}$ can be determined by taking the exterior derivative 
of the one-form $\omega^0$, i.e., 
\begin{equation}
d\omega^0=du^i\wedge dx^i+dr\wedge dS +d\lambda^k\wedge d\mu^k\equiv 
\frac{1}{2}{\sf K}^0_{ij} dz^i\wedge dz^j. \label{eq:5.11}
\end{equation}
Thus, the non-zero ${\sf K}^0_{ij}$ are:
\begin{equation}
{\sf K}^0_{u^i,x^i}=1,\quad {\sf K}^0_{x^i,u^i}=-1,\quad {\sf K}^0_{r,S}=1,
\quad {\sf K}^0_{S,r}=-1, \quad {\sf K}^0_{\lambda^k,\mu^k}=1. 
\label{eq:5.12}
\end{equation}
Similarly, 
$d\omega^k=d\pi_{ik}\wedge dx^i$, gives the non-zero ${\sf K}^k_{ij}$ as:
\begin{equation}
{\sf K}^k_{\pi_{ik},x^i}=1\quad \hbox{and}\quad {\sf K}^k_{x^i,\pi_{ik}}=-1. \label{eq:5.13}
\end{equation}
Using the notation (\ref{eq:5.4}) in (\ref{eq:5.12})-(\ref{eq:5.13})  we 
obtain the non-zero coefficients:
\begin{align}
&{\sf K}^0_{4,1}={\sf K}^0_{5,2}={\sf K}^0_{6,3}={\sf K}^0_{17,16}=1,\nonumber\\ 
&{\sf K}^0_{28,27}={\sf K}^0_{30,29}=\ldots =1,  \nonumber\\
&{\sf K}^1_{7,1}={\sf K}^2_{8,1}={\sf K}^3_{9,1}=1,\nonumber\\
&{\sf K}^1_{10,2}={\sf K}^2_{11,2}={\sf K}^3_{12,2}=1,\nonumber\\
&{\sf K}^1_{13,3}={\sf K}^2_{14,3}={\sf K}^3_{15,3}=1,\label{eq:5.14}
\end{align}
where ${\sf K}^\alpha_{ba}=-{\sf K}^\alpha_{ab}$. 

Using (\ref{eq:5.14}) for the matrices ${\sf K}^\alpha_{ab}$ 
($\alpha=0,1,2,3$), and the Hamiltonian $H$  in (\ref{eq:4.20})
with Hamiltonian density $h$ of (\ref{eq:4.2}), 
the multi-symplectic equations (\ref{eq:5.1}) 
gives the de Donder-Weyl equations (\ref{eq:4.7})-(\ref{eq:4.23}). 
Thus, for example, for $i=1,2,3$, (\ref{eq:5.1}) gives the gas dynamic momentum, or 
Euler equations:
\begin{equation}
\frac{du^i}{dt}+\derv{m^k}\left( pA_{ik}\right)+\deriv{\Phi}{x^i}=0, \quad i=1,2,3. \label{eq:5.15}
\end{equation}
For $i=4,5,6$, (\ref{eq:5.1}) gives the Lagrangian map equations:
\begin{equation}
\deriv{\bf x}{t}={\bf u}=\frac{\delta H}{\delta {\bf u}}.  \label{eq:5.16}
\end{equation}
For $7\leq i\leq 15$, we get the Lagrangian map equations:
\begin{equation}
\deriv{x^p}{m^q}=\frac{\delta H}{\delta \pi_{pq}}=x_{pq}, \quad 1\leq p,q\leq 3. \label{eq:5.17}
\end{equation}
For $i=16$ and $i=17$ we get the canonically conjugate equations:
\begin{equation}
\frac{dr}{dt}=-\frac{\delta H}{\delta S}=-T,\quad \frac{dS}{dt}=\frac{\delta H}{\delta r}=0. 
\label{eq:5.18}
\end{equation}
For $18\leq i\leq 26$, we obtain:
\begin{equation}
0=\frac{\delta H}{\delta x_{pq}}=\pi_{pq}-pA_{pq},\quad 1\leq p,q\leq 3, \label{eq:5.18a}
\end{equation}
where $A_{pq}=\hbox{cofac}(x_{pq})$ is the cofactor of $x_{pq}$.  For $i>26$, we obtain 
the equations:
\begin{equation}
\frac{d\mu^k}{dt}=\frac{\delta H}{\delta \lambda^k}=0, \quad 
\frac{d\lambda^k}{dt}=-\frac{\delta H}{\delta \mu^k}=0, \label{eq:5.18b}
\end{equation}
which are the equations for the $\mu^k$ and $\lambda^k$  Clebsch variables.
Thus, the multi-symplectic equations (\ref{eq:5.1}) are equivalent to the de 
Donder-Weyl equations 
(\ref{eq:4.7})-(\ref{eq:4.23}). 

\subsection{Pullback conservation laws}
From  Hydon (2005) (see also Webb et al. (2014c)), the multi-symplectic 
system (\ref{eq:5.1}) admits pullback conservation laws associated with the 
Legendre transformation
for the system. The pullback conservation laws have the form:
\begin{equation}
D_\alpha\left(L^\alpha_j z^j_{,\beta}-L\delta^\alpha_\beta\right)=0, 
\quad 0\leq\alpha,\beta\leq n,  
\label{eq:5.19}
\end{equation}
where the independent variables are $q^\alpha=(t,m^1,m^2,m^3)$ and $D_{\alpha}
\equiv \partial/\partial m^{\alpha}$ (note $D_t$ is the Lagrangian time derivative 
moving with the flow). The pullback conservation 
laws can also be derived by using Noether's theorem for the system (see e.g. Webb et al. (2014c) 
and Appendix C). In (\ref{eq:5.19}), the Lagrangian density 
\begin{equation}
L=\frac{1}{2} u^2-e(\rho, S)-\Phi({\bf x}), \label{eq:5.20}
\end{equation}
which is the Lagrangian density (\ref{eq:3.7}) without constraints. 

For $\beta=0$, the pullback conservation law (\ref{eq:5.19}) becomes:
\begin{equation}
\deriv{I^0}{t}+\deriv{I^j}{m^j}=0, \label{eq:5.21}
\end{equation}
where
\begin{equation}
I^0=L^0_k z^k_{,0}-L, \quad I^j=L^j_k z^k_{,0}, \label{eq:5.22}
\end{equation}
 To evaluate the conserved density $I^0$ and conserved current $I^j$, 
 note from the expressions for $\omega^0$ and $\omega^k$ in (\ref{eq:5.7}) and (\ref{eq:5.8})
that:
\begin{align}
I^0=&L^0_k \deriv{z^k}{t}-L=u^k\deriv{x^k}{t}+r \frac{dS}{dt} +\lambda^k \frac{d\mu^k}{dt}
-\left(\frac{1}{2} u^2-e(\rho,S)-\Phi({\bf x})\right)\nonumber\\
=&\frac{1}{2} u^2+e(\rho,S)+\Phi({\bf x}), \label{eq:5.23}
\end{align}
 Thus, $I_0$ is 
the kinetic plus potential energy of the fluid. Similarly, $I^j$ from (\ref{eq:5.22}) 
and (\ref{eq:5.7}) reduces to:
\begin{equation}
I^j=L^j_k \deriv{z^k}{t}=pA_{kj} u^k. \label{eq:5.24}
\end{equation}
 Using (\ref{eq:5.23}) and (\ref{eq:5.24}), 
the pullback conservation law (\ref{eq:5.21}) reduces to:
\begin{equation}
\frac{D}{Dt}\left(\frac{1}{2} u^2+e(\rho,S)+\Phi({\bf x})\right)
+\derv{m^j}\left(p A_{kj} u^k\right)=0, \label{eq:5.25}
\end{equation}
is the Lagrangian total energy equation, where $D/Dt=\partial_t+{\bf u}{\bf\cdot}\nabla$ 
is the time derivative following the flow.  By noting $\partial A_{kj}/\partial m^j=0$ 
 and using $A_{kj}=Jy_{jk}\equiv y_{jk}/\rho$ 
where $y_{jk}=\partial m^j/\partial x^k$, (\ref{eq:5.25}) reduces to the equation:
\begin{equation}
\derv{t}\left(\frac{1}{2}u^2+e(\rho,S)+\Phi({\bf x})\right) 
+{\bf u}{\bf\cdot}\nabla\left(\frac{1}{2}u^2+e(\rho,S)+\Phi({\bf x})\right)
+\frac{1}{\rho}\nabla{\bf\cdot}(p{\bf u})=0. \label{eq:5.26}
\end{equation}
Then, using mass continuity equation (\ref{eq:2.1}), 
(\ref{eq:5.26}) reduces to the Eulerian energy equation:
\begin{equation}
\derv{t}\left(\frac{1}{2}\rho u^2+\varepsilon+\rho\Phi\right)
+\nabla{\bf\cdot}\left[\rho{\bf u}\left(\frac{1}{2}u^2+w+\Phi({\bf x})\right)\right]=0, \label{eq:5.27}
\end{equation}
where $w=(\varepsilon+p)/\rho$ is the gas enthalpy. 

Next, consider the pullback conservation laws (\ref{eq:5.19}) for the case $\beta=i$. 
In this case (\ref{eq:5.19}) reduce to equations of the form:
\begin{equation}
\derv{t}T^{0i}+\derv{m^k}\left(T^{ki}\right)=0, \label{eq:5.28}
\end{equation}
where
\begin{equation}
T^{0i}=L^0_j z^j_{,i},\quad T^{ki}=L^k_j z^j_{,i}-L\delta^k_i, \label{eq:5.29}
\end{equation}
are the density and flux respectively. 
Using the differential forms (\ref{eq:5.7}) and (\ref{eq:5.8}) 
for $\omega^0$ and $\omega^k$, we obtain:
\begin{equation}
T^{0i}=u^j x_{ji}+r\deriv{S}{m^i}+\lambda^k \deriv{\mu^k}{m^i},\quad 
T^{ki}=\left(w+\Phi-\frac{1}{2} u^2\right)\delta^{ki}, \label{eq:5.30}
\end{equation} 
for the density and flux.  The conservation laws (\ref{eq:5.28}) reduce to:
\begin{equation}
\derv{t}\left(u^j x_{ji}+r\deriv{S}{m^i}+\lambda^k \deriv{\mu^k}{m^i}\right) 
+\derv{m^i}\left(w+\Phi-\frac{1}{2} u^2\right)=0. \label{eq:5.31}
\end{equation}
However, 
\begin{equation}
\derv{t}\left(\lambda^k \deriv{\mu^k}{m^i}\right)=\frac{d\lambda^k}{dt}\deriv{\mu^k}{m^i}+
\lambda^k \derv{m^i}\left(\frac{d\mu^k}{dt}\right)=0, \label{eq:5.31a}
\end{equation}
because $d\lambda^k/dt=0$ and $d\mu^k/dt=0$ by (\ref{eq:5.18b}). Thus, (\ref{eq:5.31}) 
may also be written as:
\begin{equation}
\derv{t}\left(u^j x_{ji}+r\deriv{S}{m^i}\right)
+\derv{m^i}\left(w+\Phi-\frac{1}{2} u^2\right)=0. \label{eq:5.31b}
\end{equation}
which shows that (\ref{eq:5.31})does not depend on the gauge potentials $\mu^k$ and $\lambda^k$.
It is interesting to note, that
\begin{equation}
T^{0i}=\deriv{x^j}{m^i}\left(u^j+r\deriv{S}{x^j}+\lambda^k\deriv{\mu^k}{x^j}\right)\equiv 
\deriv{\phi}{m^i}-\beta_s\deriv{A_s}{m^i}, \label{eq:5.31c}
\end{equation}
where we have used the Eulerian Clebsch representation for ${\bf u}$:
\begin{equation}
{\bf u}=\nabla\phi -r\nabla S-\lambda^k\nabla\mu^k-\beta_s \nabla_i A_s. \label{eq:5.31d}
\end{equation}
Here the $-\beta_s\nabla A^s$ term arises in the conversion of the Lagrangian variational 
principle to an Eulerian variational principle (e.g. Fukugawa and Fujitani (2010)). 
Since $d\beta_s/dt=0$ and $d A^s/dt=0$ 
 (\ref{eq:5.31}) reduces to:
\begin{equation}
\derv{m^i}\left(\frac{d\phi}{dt}+w+\Phi-\frac{1}{2} u^2\right)=0. \label{eq:5.31e}
\end{equation}
Equation (\ref{eq:5.31e}) can be integrated with respect to the $m^i$ to obtain 
Bernoulli's equation:
\begin{equation}
\frac{d\phi}{dt}+w+\Phi-\frac{1}{2} u^2=f(t), \label{eq:5.31f}
\end{equation}
where $f(t)$ is the `integration constant'. Thus, the pullback conservation law
(\ref{eq:5.31}) is equivalent to Bernoulli's equation.

By using (\ref{eq:2.7}), (\ref{eq:5.31}) reduces to:
\begin{equation}
x_{ji}\left(\frac{du^j}{dt}+\deriv{\Phi}{x^j}+\frac{1}{\rho}\deriv{p}{x^j}\right)=0, \label{eq:5.32}
\end{equation}
which implies the Euler momentum equation (\ref{eq:2.2}) provided $\det(x_{ij})\neq 0$. 
The conservation laws
(\ref{eq:5.31}) are due to Noether's theorem, and 
 translation invariance of the action with respect to the 
 mass coordinates $m^i$. This is different than translation invariance 
of the action with respect to the  
$x^i$ which implies momentum conservation (e.g. for cases where 
$\Phi({\bf x})=0$).  

The Eulerian conservation law corresponding to (\ref{eq:5.21}) 
is:
\begin{equation}
\deriv{F^0}{t}+\deriv{F^j}{x^j}=0, \label{eq:5.33}
\end{equation}
where
\begin{equation}
F^0=\frac{I^0}{J},\quad F^j=\frac{u^jI^0+x_{jk} I^k}{J}, \label{eq:5.34}
\end{equation}
are the Eulerian density and flux respectively (Padhye 1998). 
Applying the result (\ref{eq:5.33})-(\ref{eq:5.34}), 
the conservation law (\ref{eq:5.31}) reduces to its Eulerian form:
\begin{equation}
\derv{t}\left[\rho x_{ji}\left(u^j+r\deriv{S}{x^j}\right)\right]
+\derv{x^j}\left[x_{ki}\rho u^j\left(u^k+r\deriv{S}{x^k}\right)
+\rho\left(w+\Phi-\frac{1}{2}u^2\right)x_{ji}\right]=0. \label{eq:5.35}
\end{equation}

In Appendix D, we give an independent verification of (\ref{eq:5.35}) by using  a Clebsch  
variational principle, in which mass conservation, entropy advection and 
the Lin constraint are imposed by using Lagrange multipliers. 

Webb (2015) studied  Lagrangian, ideal, compressible gas dynamics 
in one Cartesian space dimension, for the case  $\Phi(x)=0$. 
 (\ref{eq:5.25}) and (\ref{eq:5.27}) then reduce to: 
\begin{align}
&\derv{t}\left(\frac{1}{2}u^2+\frac{\varepsilon}{\rho}\right) +\derv{m}(pu)=0, \label{eq:5.44}\\
&\derv{t}\left(\frac{1}{2}\rho u^2+\varepsilon\right) 
+\derv{x}\left[\rho u\left(\frac{1}{2}u^2+w\right)\right]=0, \label{eq:5.45}
\end{align}
corresponding to the $\beta=0$ (i.e. the energy 
conservation equations). The mass translation invariant laws (\ref{eq:5.31}) 
and (\ref{eq:5.35}) reduce to:
\begin{align}
&\derv{t}\left(\frac{u}{\rho}+r\deriv{S}{m}\right)+\derv{m}\left(w-\frac{1}{2}u^2\right)=0, 
\label{eq:5.46}\\
&\derv{t}\left(u+r\deriv{S}{x}\right)+\derv{x}\left(w+\frac{1}{2} u^2+ur \deriv{S}{x}\right)=0. 
\label{eq:5.47}
\end{align}

Equation (\ref{eq:5.47}) implies the momentum equation for 1D gas dynamics 
with $\Phi(x)=0$. For an isobaric  gas, 
i.e. $p=p(\rho)$ and $\varepsilon=\varepsilon(\rho)$, $S$  is constant throughout the flow, and 
$\partial S/\partial x=0$, and (\ref{eq:5.47}) is then a local conservation law. 
However, for  $p=p(\rho,S)$ and $\varepsilon=\varepsilon(\rho,S)$, 
 (\ref{eq:5.47}) is a nonlocal conservation 
law, depending on the nonlocal variable $r$. Using (\ref{eq:5.41}) 
 $r(x,t)$ can be expressed in the form:
\begin{equation}
r(x,t)=-\int_0^t \bar{T} (m,t')\ dt'+r_0(m), \label{eq:5.48}
\end{equation}
where $T(x,t)=\bar{T}(m,t)$ and $r_0(m)=r(m,0)$. From (\ref{eq:5.48}), 
$r(x,t)$ depends on the time integrated temperature associated with the Lagrange 
label $m$. Thus, $r$  is a nonlocal variable. 

\subsection{Symplecticity conservation laws}

In standard Hamiltonian mechanics in which the time $t$ is 
the evolution variable, the phase space
element $\kappa=dp_i\wedge dq^i=d(p_idq^i)$ is conserved following the Hamiltonian flow, i.e. 
$d\kappa/dt=0$. The generalization of this conservation law for multi-symplectic systems is 
the phase space conservation equation:
\begin{equation}
\kappa^\alpha_{,\alpha}=0\quad \hbox{where}\quad \kappa^\alpha=d\omega^\alpha, 
\quad \omega^\alpha=L^\alpha_j dz^j. \label{eq:5.49}
\end{equation}
From (\ref{eq:5.49}) and (\ref{eq:5.9})-(\ref{eq:5.10}) the fundamental two-form $\kappa^\alpha$ has 
the form:
\begin{equation}
\kappa^\alpha=d\omega^\alpha=\frac{1}{2}{\sf K}^\alpha_{ij}dz^i\wedge dz^j\quad\hbox{where}\quad
{\sf K}^\alpha_{ij}=\deriv{L^\alpha_j}{z^i}-\deriv{L^\alpha_i}{z^j}. \label{eq:5.50}
\end{equation}
The matrix ${\sf K}^\alpha_{ij}$ is skew symmetric in the 2 lower indices. 

The symplecticity conservation law $\kappa^{\alpha}_{,\alpha}=0$ when pulled back to the base manifold
in which $q^\alpha$ ($\alpha=0,1,2,3$), are the independent variables is:
\begin{equation}
D_\alpha\left(\frac{1}{2}{\sf K}^\alpha_{ij}dz^i\wedge dz^j\right)
=\frac{1}{2}D_\alpha\left({\sf K}^{\alpha}_{ij}\deriv{z^i}{q^\beta}\deriv{z^j}{q^\gamma} 
dq^\beta\wedge dq^\gamma\right)=0.
\label{eq:5.51}
\end{equation}
Since the forms $dq^\beta\wedge dq^\gamma$ (for $\beta<\gamma$) 
are independent, (\ref{eq:5.51}) is satisfied if 
\begin{equation}
D_\alpha\left({\sf K}^\alpha_{ij}\deriv{z^i}{q^\beta}\deriv{z^j}{q^\gamma}\right)
=0\quad \hbox{where} \quad \beta<\gamma. \label{eq:5.52}
\end{equation}
We take ${\bf q}=(t,m^1,m^2,m^3)^T$ where the 
$\{m^i\}$ are  the mass coordinates. 

The symplecticity laws (\ref{eq:5.52}) can be obtained by cross differentiation of the 
pullback laws (\ref{eq:5.19}). Using (\ref{eq:5.19}) we obtain:
\begin{equation}
D_\gamma G_\beta-D_\beta G_\gamma=D_\alpha\left({\sf K}^\alpha_{ij} z^i_{,\gamma}z^j_{,\beta}\right),
\label{eq:5.53}
\end{equation}
where
\begin{equation}
G_\beta=D_\alpha\left(L^\alpha_j z^i_{,\beta}-L\delta^\alpha_{\beta}\right). \label{eq:5.54}
\end{equation}
 Here $D_\alpha\equiv \partial/\partial m^\alpha$. 
 The equations 
$G_\beta=0$ ($\beta=0,1,2,3$) are the pullback laws (\ref{eq:5.19}). Thus, 
the symplecticity conservation laws are compatibility conditions for the 
 laws $G_\beta=0$. 

The symplecticity laws 
 can be obtained from the pullback of (\ref{eq:5.49}) rather than 
using the symplecticity  law form  (\ref{eq:5.52}).  
From (\ref{eq:5.7}) and (\ref{eq:5.8}) we obtain:
\begin{align}
\kappa^0=&d\omega^0=du^i\wedge dx^i+dr\wedge dS+d\lambda^k\wedge d\mu^k, \label{eq:5.55}\\
\kappa^k=&d\omega^k=d\pi_{ik}\wedge dx^i. \label{eq:5.56}
\end{align}
Using the pullback operation  
  (\ref{eq:5.55}) gives:
\begin{align}
\psi^*\kappa^0=&
\left(\deriv{u^i}{t} dt+\deriv{u^i}{m^j} dm^j\right) 
\wedge \left(\deriv{x^i}{t} dt+\deriv{x^i}{m^s} dm^s\right)\nonumber\\
&+\left(\deriv{r}{t} dt+\deriv{r}{m^j} dm^j\right)
\wedge\left(\deriv{S}{t}dt+\deriv{S}{m^s} dm^s\right)\nonumber\\
&+\left(\deriv{\lambda^k}{t} dt+\deriv{\lambda^k}{m^j}dm^j\right)\wedge
\left(\deriv{\mu^k}{t} dt+\deriv{\mu^k}{m^s}dm^s\right)\nonumber\\
=&\biggl[\deriv{u^i}{t} \deriv{x^i}{m^s}-\deriv{u^i}{m^s} \deriv{x^i}{t} 
+\deriv{r}{t}\deriv{S}{m^s}-\deriv{S}{t}\deriv{r}{m^s}\biggr] dt\wedge dm^s\nonumber\\
&+\sum_{j<s}\biggl(\deriv{u^i}{m^j}\deriv{x^i}{m^s}-\deriv{u^i}{m^s}\deriv{x^i}{m^j}
+\deriv{r}{m^j}\deriv{S}{m^s}-\deriv{r}{m^s}\deriv{S}{m^j}\nonumber\\
&+\deriv{\lambda^k}{m^j}\deriv{\mu^k}{m^s}-\deriv{\lambda^k}{m^s}\deriv{\mu^k}{m^j}\biggr)
 dm^j\wedge dm^s. 
\label{eq:5.57} 
\end{align}
where $\psi^*$ denotes the pullback map to the base manifold. 
Similarly, we find:
\begin{align}
\psi^*\kappa^k=&\left(\deriv{\pi_{ik}}{t}\deriv{x^i}{m^s}-\deriv{\pi_{ik}}{m^s}\deriv{x^i}{t} \right)
dt\wedge dm^s\nonumber\\
&+\sum_{j<s}\left(\deriv{\pi_{ik}}{m^j}\deriv{x^i}{m^s}-\deriv{\pi_{ik}}{m^s}\deriv{x^i}{m^j}
\right) dm^j\wedge dm^s. \label{eq:5.58}
\end{align}

Using (\ref{eq:5.57}) and (\ref{eq:5.58}) in  (\ref{eq:5.49}) 
 we obtain the symplecticity conservation 
laws:
\begin{equation}
\derv{t}\left[\frac{\partial(u^i,x^i)}{\partial(t,m^s)} +\frac{\partial(r,S)}{\partial(t,m^s)}\right]
+\derv{m^k}\left[\frac{\partial(\pi_{ik},x^i)}{\partial(t,m^s)}\right]=0, \label{eq:5.59}
\end{equation}
corresponding to the $dt\wedge dm^s$ balance. In (\ref{eq:5.59}) 
\begin{equation}
\frac{\partial(f,g)}{\partial(x,y)}=f_xg_y-f_yg_x, \label{eq:5.60}
\end{equation}
denotes the Jacobian of the functions $f$ and $g$ with respect to $x$ and $y$. Similarly,
equating the $dm^j\wedge dm^s$ terms ($j<s$) in the equations $\kappa^\alpha_{,\alpha}=0$, 
we obtain the symplecticity conservation laws:
\begin{equation}
\derv{t}\left[\frac{\partial(u^i,x^i)}{\partial(m^j,m^s)} 
+\frac{\partial(r,S)}{\partial(m^j,m^s)}+\frac{\partial(\lambda^k,\mu^k)}{\partial(m^j,m^s)}
\right]
+\derv{m^k}\left[\frac{\partial(\pi_{ik},x^i)}{\partial(m^j,m^s)}\right]=0. \label{eq:5.61}
\end{equation}
There are thus, 6 symplecticity conservation laws, 3 in (\ref{eq:5.59}) and 3 in (\ref{eq:5.61}). 

Bridges et al. (2005) elucidate the connection between vorticity and symplecticity in 
fluid dynamics using a range of fluid models: the incompressible fluid, the barotropic fluid, 
and the shallow water equations used in geophysical fluid dynamics, including the effects of 
the Coriolis force and investigated conserved invariants due to 
fluid relabelling symmetries and Noether's second theorem.
Hydon and Mansfield
(2011) give a clear exposition of Noether's second theorem with applications. 

\begin{proposition}\label{5.1}
The symplecticity conservation law (\ref{eq:5.59}) reduces to the law:
\begin{equation}
\derv{m^s}\biggl\{\frac{d}{dt}\left(\frac{1}{2} u^2+e+\Phi({\bf x})\right)
+\derv{m^j}\left(p A_{kj} u^k\right)\biggr\}=0, \label{eq:5.62a}
\end{equation}
where $e=\varepsilon/\rho$ is the internal energy of the gas per unit mass. 
Equation (\ref{eq:5.62a}) is the derivative with respect to $m^s$ of 
the co-moving total energy equation (\ref{eq:5.25}). 
\end{proposition}

\begin{proof}
First note that (\ref{eq:5.59}) is equivalent to the conservation law:
\begin{equation}
\deriv{D}{t}+\deriv{F^k}{m^k}=0, \label{eq:5.63a}
\end{equation}
where
\begin{equation}
D= \frac{\partial(u^i,x^i)}{\partial(t,m^s)}+
\frac{\partial(r,S)}{\partial(t,m^s)}, \quad
F^k=\frac{\partial(\pi_{ik},x^i)}{\partial(t,m^s)}. \label{eq:5.64a}
\end{equation}
Using (\ref{eq:2.2}) and (\ref{eq:2.7}) and noting 
$\partial S/\partial t\equiv dS/dt=0$, we obtain:
\begin{align}
D=&\left(T\deriv{S}{x^i}-\deriv{w}{x^i}-\deriv{\Phi}{x^i}\right)x_{is}
-u^i\deriv{u^i}{m^s} 
+\frac{dr}{dt}\deriv{S}{m^s}\nonumber\\
=&\left(\frac{dr}{dt}+T\right)\deriv{S}{m^s}-\derv{m^s}\left(w+\Phi+\frac{1}{2}u^2\right)
\equiv -\derv{m^s}\left(w+\Phi+\frac{1}{2}u^2\right). \label{eq:5.66a}
\end{align}
Similarly, we obtain:
\begin{equation}
F^k=\derv{t}\left(\pi_{ik} x_{is}\right)
-\derv{m^s}\left(\pi_{ik}u^i\right)
=\derv{t}\left(p\tau\delta^k_s\right)-\derv{m^s}\left(p A_{ik} u^i\right), \label{eq:5.67a}
\end{equation}
where $\tau=1/\rho=J$. Using (\ref{eq:5.66a}) and (\ref{eq:5.67a}) in (\ref{eq:5.59})
gives:
\begin{equation}
\deriv{D}{t}+\deriv{F^k}{m^k}=-\derv{m^s}
\biggl\{\frac{d}{dt}\left(\frac{1}{2} u^2+e+\Phi({\bf x})\right)
+\derv{m^k}\left(p A_{ik} u^i\right)\biggr\}=0. \label{eq:5.68a}
\end{equation}
This proves (\ref{eq:5.62a}).
\end{proof}

 Below, we investigate the symplecticity laws (\ref{eq:5.61})
for the case where $j=\alpha$ and $s=\beta$ are fixed values of the indices $j$  and 
$s$ in (\ref{eq:5.61}).
\begin{proposition}\label{5.2}
The symplecticity conservation laws (\ref{eq:5.61}) can be written as:
\begin{equation}
\frac{dI^0_{\alpha\beta}}{dt}=0,\quad (1\leq\alpha < \beta\leq 3). 
\label{eq:5.62}
\end{equation}
(\ref{eq:5.62}) are a consequence of the symplecticity conservation laws:
\begin{equation}
\deriv{I^0_{\alpha\beta}}{t}+\deriv{I^k_{\alpha\beta}}{m^k}=0, \label{eq:5.62ab}
\end{equation}
in which $\partial I^k_{\alpha\beta}/\partial m^k=0$, and
\begin{align}
I_{\alpha\beta}^0=&{\bf e}_\alpha\times {\bf e}_\beta {\bf\cdot}
\left(\boldsymbol{\omega}
+\nabla r\times\nabla S+\nabla\lambda^k\times\nabla\mu^k\right), 
\quad (1\leq \alpha <\beta\leq 3), \label{eq:5.63ab}\\
I_{\alpha\beta}^k=&\derv{m^\alpha}\left(p J\delta^k_\beta\right)
-\derv{m^\beta}\left(p J\delta^k_\alpha\right),
\quad J=\det(x_{ij})=\frac{1}{\rho},\label{eq:5.63ac}\\
\boldsymbol{\omega}=&\nabla\times{\bf u},\quad {\bf e}_\mu=\deriv{\bf x}{m^\mu}, \quad 
(1\leq\mu\leq 3). \label{eq:5.63ad}
\end{align}
The term involving $\nabla\lambda^k\times\nabla \mu^k$ in (\ref{eq:5.63ab}) 
does not contribute to the conservation law (\ref{eq:5.62}). Thus, we may take:  
\begin{equation}
I^0_{\alpha\beta}={\bf e}_\alpha\times {\bf e}_\beta {\bf\cdot}\left(\boldsymbol{\omega}
+\nabla r\times\nabla S\right). \label{eq:5.65b}
\end{equation}
\end{proposition}

\begin{proof}
From (\ref{eq:5.63ac}):
\begin{equation}
\deriv{I^k_{\alpha\beta}}{m^k}=\derv{m^\beta}
\left[\derv{m^\alpha}\left(\frac{p}{\rho}\right)\right]
-\derv{m^\alpha}\left[\derv{m^\beta}\left(\frac{p}{\rho}\right)\right]
=0. \label{eq:5.68}
\end{equation}

To prove (\ref{eq:5.62}), note that the conserved density from (\ref{eq:5.61}) has the form:
\begin{equation}
D=D_1+D_2+D_3\quad \hbox{where}\quad D_1=\frac{\partial(u^i,x^i)}{\partial(m^\alpha,m^\beta)}, \quad
D_2=\frac{\partial(r,S)}{\partial(m^\alpha,m^\beta)},
\quad D_3=\frac{\partial(\lambda^k,\mu^k)}{\partial(m^\alpha,m^\beta)}. \label{eq:5.70}
\end{equation}
Using the Lagrangian map, $x^i=x^i({\bf m},t)$ and assuming $J\neq 0$, we obtain:
\begin{equation}
D_1=\Omega_{ik} x_{k\alpha} x_{i\beta}\quad\hbox{where}\quad \Omega_{ik}=\deriv{u^i}{x^k}-\deriv{u^k}{x^i}, 
\label{eq:5.71}
\end{equation}
is a skew symmetric matrix (tensor) which is related to the fluid vorticity 
$\boldsymbol{\omega}=\nabla\times{\bf u}$.  Using the hat map (Holm (2008), Vol. 2):
we obtain:. 
\begin{equation}
\Omega_{ij}=-\varepsilon_{ijk}\omega^k,\quad \omega^k=-\frac{1}{2}\varepsilon_{kij} \Omega_{ij}. 
\label{eq:5.74}
\end{equation}
Using (\ref{eq:5.74}) in (\ref{eq:5.71}), we obtain:
\begin{equation}
D_1={\bf e}_{\alpha}\times{\bf e}_{\beta} {\bf\cdot}\boldsymbol{\omega}\equiv 
\frac{\partial(u^i,x^i)}{\partial(m^\alpha,m^\beta)}. \label{eq:5.76}
\end{equation}

A similar calculation for the second term $D_2$ in (\ref{eq:5.70}) gives:
\begin{equation}
D_2=\frac{\partial(r,S)}{\partial(m^\alpha,m^\beta)}=\Phi_{ik} x_{i\alpha}x_{k\beta}\quad\hbox{where}\quad 
\Phi_{ik}=\frac{\partial(r,S)}{\partial( x^i,x^k)}. \label{eq:5.77}
\end{equation}
The matrix $\Phi_{ik}$ is skew symmetric with $\Phi_{ki}=-\Phi_{ik}$. 
We find:
\begin{equation}
D_2=\frac{\partial(r,S)}{\partial(m^\alpha, m^\beta)}={\bf e}_\alpha\times {\bf e}_\beta
{\bf\cdot}(\nabla r\times\nabla S). \label{eq:5.81}
\end{equation}
 Similarly,
\begin{equation}
D_3=\frac{\partial(\lambda^k,\mu^k)}{\partial(m^\alpha, m^\beta)}
={\bf e}_\alpha\times {\bf e}_\beta{\bf\cdot}(\nabla \lambda^k\times\nabla \mu^k). 
\label{eq:5.81a}
\end{equation}
It is straightforward to verify that $d/dt(D_3)=0$ using (\ref{eq:5.81a}), because 
$d\lambda^k/dt=0$ and $d\mu^k/dt=0$.
Using (\ref{eq:5.76}) for $D_1$ and (\ref{eq:5.81}) for $D_2$ 
gives the result (\ref{eq:5.63ab}) 
for  $I^0_{\alpha\beta}$. Since $dD_3/dt=0$, (\ref{eq:5.63ab}) reduces to (\ref{eq:5.65b}).  

The conserved flux $I^k_{\alpha\beta}$ in (\ref{eq:5.61})
and (\ref{eq:5.62}) is given by:
\begin{equation}
I^k_{\alpha\beta}=\frac{\partial(\pi_{ik},x^i)}{\partial(m^\alpha,m^\beta)}
\equiv \derv{m^\alpha}\left(pA_{ik}x_{i\beta}\right)
-\derv{m^\beta}\left(pA_{ik}x_{i\alpha}\right). \label{eq:5.82}
\end{equation}
By noting that $A_{ik}x_{i \beta}=J\delta^k_\beta$, (\ref{eq:5.82}) reduces to 
the expression (\ref{eq:5.63ac}) for $I^k_{\alpha\beta}$. Note that 
$\partial I^k_{\alpha\beta}/\partial m^k=0$ in (\ref{eq:5.68}) 
follows from (\ref{eq:5.63ac}), 
which implies that $I^0_{\alpha\beta}$ is a conserved density that is advected with 
the flow. This completes the proof. 
\end{proof}

\begin{proposition}\label{5.3}
The conservation laws (\ref{eq:5.62}) imply
\begin{equation}
\frac{d}{dt}\left(\frac{\Omega^\gamma}{\rho}\right)=0. \label{eq:5.83}
\end{equation}
If $\Psi({\bf m})$ is a scalar advected with the flow, (\ref{eq:5.83}) implies
\begin{equation}
 \frac{d}{dt}\left(\frac{\boldsymbol{\Omega}{\bf\cdot}\nabla \Psi}
{\rho}\right)=0\quad \hbox{and}\quad 
\frac{d}{dt}\left(\frac{\boldsymbol{\omega}{\bf\cdot}\nabla S}{\rho}\right)=0, 
\label{eq:5.83a}
\end{equation}
where 
\begin{equation}
\boldsymbol{\Omega}=\boldsymbol{\omega}+\nabla r\times\nabla S
=\Omega^\gamma {\bf e}_\gamma. \label{eq:5.84}
\end{equation}
The conservation law (\ref{eq:5.83a}) is equivalent to Ertel's theorem, i.e. $dq/dt=0$ where 
$q=\boldsymbol{\Omega}{\bf\cdot}\nabla \Psi/\rho$ is the potential vorticity. In the particular 
case $\Psi=S$, then $q=\boldsymbol{\Omega}{\bf\cdot}\nabla\Psi/\rho\equiv 
\boldsymbol{\omega}{\bf\cdot}\nabla S/\rho$. For $\Psi\neq S$, 
\begin{equation}
 q=q_c+
\frac{\nabla r\times\nabla S{\bf\cdot}\nabla\Psi}{\rho}\quad\hbox{where}
\quad q_c=\frac{\boldsymbol{\omega}{\bf\cdot}\nabla\Psi}{\rho} \label{eq:5.84a}
\end{equation} 
is the classical potential vorticity. 
Thus, $q$  differs from 
$q_c$  
if $\nabla r\times\nabla S{\bf\cdot}\nabla\Psi\neq 0$.
\end{proposition}

\begin{proof}
From (\ref{eq:5.62}) and (\ref{eq:3.11}):
\begin{equation}
\frac{d}{dt}\left({\bf e}_\alpha\times{\bf e}_\beta{\bf\cdot}\boldsymbol{\Omega}\right)=\frac{d}{dt}
\left(\epsilon_{\alpha\beta\gamma} {\bf e}^\gamma {\bf\cdot} \frac{\boldsymbol{\Omega}}{\rho}\right)
=\epsilon_{\alpha\beta\gamma}\frac{d}{dt}\left(\frac{\Omega^\gamma}{\rho}\right)=0. 
\label{eq:5.85}
\end{equation}
Equation (\ref{eq:5.85}) implies $d(\Omega^\gamma/\rho)/dt=0$ which establishes 
(\ref{eq:5.83}). Next notice that:
\begin{equation}
 \frac{\boldsymbol{\Omega}{\bf\cdot}\nabla \Psi}{\rho}
=\frac{\Omega^k}{\rho}\deriv{\Psi}{m^k}. \label{eq:5.86}
\end{equation}
Taking the Lagrangian time derivative of (\ref{eq:5.86}) gives:
\begin{equation}
\frac{d}{dt}\left(\frac{\boldsymbol{\Omega}{\bf\cdot}\nabla \Psi}{\rho}\right)
=\frac{d}{dt}\left(\frac{\Omega^k}{\rho}\deriv{\Psi}{m^k}\right)
=\frac{d}{dt}\left(\frac{\Omega^k}{\rho}\right)\deriv{\Psi}{m^k}
+\frac{\Omega^k}{\rho} \frac{d}{dt}\left(\deriv{\Psi}{m^k}\right)=0. \label{eq:5.87}
\end{equation}
In (\ref{eq:5.87}), $d/dt(\Omega^k/\rho)=0$ by (\ref{eq:5.83}) 
and $d/dt(\partial \Psi/\partial m^k)=(\partial/\partial m^k)(d\Psi/dt)=0$. 
This proves the generalized form of Ertel's theorem involving the nonlocal Clebsch potential $r$. 
\end{proof}

\section{Lie dragging and Noether's second theorem}
In this section we provide alternative approaches to interpret the multi-symplecticity 
conservation laws associated with fluid vorticity developed in section 5.  The concept 
of Lie dragging of geometrical objects, e.g. vectors fields, differential forms, tensors 
etc., by a vector field ${\bf V}$ is described by Schutz (1980). This describes 
the rate of change of a geometrical object ${\bf G}$ in the direction of the vector field 
${\bf V}$, which is written as ${\cal L}_{\bf V}({\bf G})\equiv d ({\bf G})/d\epsilon$
where ${\cal L}_{\bf V}$ denotes the Lie derivative with respect to the vector field 
${\bf V}$. This is the directional derivative $d/d \epsilon$ along a curve, 
with parameter $\epsilon$ and  with tangent vector ${\bf V}$, 
 corresponding to a Lie symmetry of the system. 
In order to compare like quantities, it is necessary to parallel transport
${\bf G}$ at the point along the curve, with tangent vector ${\bf V}$ 
back to the initial point ($\epsilon=0$) of the curve with tangent vector ${\bf V}$. 
In the analysis 
below, ${\bf V}\equiv{\bf u}$ is the three dimensional fluid velocity ${\bf u}$. 
Note that both base vectors and tensor components both change in the Lie derivative. 
\subsection{Lie dragging approach}

The results of propositions \ref{5.2} and \ref{5.3} are difficult to interpret, since they 
are expressed in terms of the holonomic base vectors 
${\bf e}_\alpha=\partial {\bf x}/\partial m^\alpha$ and 
${\bf e}^\gamma=\partial m^\gamma/\partial {\bf x}$. Another way in which to interpret these 
results is to use the Lie dragging approach to advected invariants in fluid mechanics 
and magnetohydrodynamics (MHD) used by Tur and Yanovsky (1993) and Webb et al. (2014a). 
From Webb et al. 
(2014a) 
(see also Appendix D), 
the fluid velocity ${\bf u}$ can be expressed in 
the Clebsch potential form:
\begin{equation}
{\bf u}=\nabla\phi-r\nabla S-\tilde{\lambda}\nabla\mu, \label{eq:fr1}
\end{equation}
where the usual Clebsch variables are $\beta=r\rho$ and $\lambda=\rho\tilde{\lambda}$. The variables 
$r$ and $\tilde{\lambda}$ and $\mu$ satisfy the equations:
\begin{equation}
\frac{d\tilde{\lambda}}{dt}=\frac{d\mu}{dt}=0,\quad \frac{dr}{dt}=-T, \label{eq:fr2}
\end{equation}
where $d/dt=\partial t+{\bf u}{\bf\cdot}\nabla$ is the Lagrangian time derivative.
We introduce the velocity field ${\bf v}$ and $\boldsymbol{\Omega}=\nabla\times {\bf v}$:
\begin{align}
{\bf v}=&{\bf u}+r\nabla S-\nabla\phi\equiv -\tilde{\lambda}\nabla\mu, \label{eq:fr3}\\
\boldsymbol{\Omega}=&\nabla\times{\bf v}=\boldsymbol{\omega}+\nabla r\times\nabla S\equiv
-\nabla\tilde{\lambda}\times\nabla\mu, \label{eq:fr4}
\end{align}
where $\boldsymbol{\omega}=\nabla\times{\bf u}$ is the total vorticity of the fluid. The vorticity 
vector $\boldsymbol{\Omega}$ represents the component of the fluid vorticity that is independent 
of the entropy gradients. 
Note that both $\tilde{\lambda}$ and $\mu$ are scalars that are advected 
(Lie dragged) by the background  flow. The one form:
\begin{equation}
\alpha={\bf v}{\bf\cdot}d{\bf x}=-\tilde{\lambda}\nabla\mu{\bf\cdot}d{\bf x}\equiv -\tilde{\lambda}d\mu,
\label{eq:fr5}
\end{equation}
is Lie dragged by the background flow ${\bf u}$.   

\begin{proposition}\label{5.4}
The one-forms $\alpha={\bf v}{\bf\cdot}d{\bf x}$ and the one form 
$\gamma=d\mu=\nabla\mu{\bf\cdot}d{\bf x}$ are advected, scalar invariants moving  
with the flow (e.g.  Webb et al. (2014a)): 
\begin{equation}
\left(\derv{t}+ {\cal L}_{\bf u}\right)\alpha=\left(\deriv{\bf v}{t}
-{\bf u}\times (\nabla\times {\bf v}) +\nabla({\bf u}{\bf\cdot}{\bf v})\right){\bf\cdot}d{\bf x}=0, 
\label{eq:fr6}
\end{equation}
where ${\cal L}_{\bf u}={\bf u}{\bf\cdot}\nabla=u^i\partial/\partial x^i$ is the Lie derivative 
with respect ${\bf u}$. 
The two-form:
\begin{equation}
\beta=d\alpha=(\nabla\times{\bf v}){\bf\cdot}d{\bf S}=\boldsymbol{\Omega}{\bf\cdot}d{\bf S}, 
\label{eq:fr7}
\end{equation}
is an advected invariant 2-form, satisfying the equation:
\begin{equation}
\left(\derv{t}+{\cal L}_{\bf u}\right) \beta=
\left(\deriv{\boldsymbol{\Omega}}{t}
-\nabla\times\left({\bf u}\times\boldsymbol{\Omega}\right)
+{\bf u}\left(\nabla{\bf\cdot}\boldsymbol{\Omega}\right)\right){\bf\cdot} d{\bf S}=0, \label{eq:fr8}
\end{equation}
which is analogous to  Faraday's equation for the magnetic field  
${\bf B}$ but with ${\bf B}$ replaced by $\boldsymbol{\Omega}$.  
Note that $\nabla{\bf\cdot}\boldsymbol{\Omega}=\nabla{\bf\cdot}\nabla\times{\bf v}=0$, 
which is analogous to $\nabla{\bf\cdot}{\bf B}=0$ in MHD. 

The 
conservation law for $\boldsymbol{\Omega}{\bf\cdot}d{\bf S}$ in (\ref{eq:fr8}) is equivalent to the symplecticity conservation 
law (\ref{eq:5.62}) in proposition (\ref{5.2}).
The vector field:
\begin{equation}
{\bf b}=\frac{\boldsymbol{\Omega}}{\rho}{\bf\cdot}\nabla, \label{eq:fr9}
\end{equation}
is Lie dragged with the flow, i.e. 
\begin{equation}
\deriv{\bf b}{t}+[{\bf u},{\bf b}]=0\quad\hbox{where}\quad [{\bf u},{\bf b}]=\left(u^j\deriv{b^i}{x^j}-b^j\deriv{u^i}{x^j}\right)\nabla_i,  \label{eq:fr10}
\end{equation}
is the left Lie bracket of the vector field of ${\bf b}$ with respect to ${\bf u}$. 

In proposition (\ref{5.3}) the equations: 
\begin{align}
\frac{\Omega^\gamma}{\rho}\equiv\frac{\boldsymbol{\Omega}{\bf\cdot}{\bf e}^\gamma}{\rho}
=&\left(\frac{\hat{\Omega}^i}{\rho}\derv{x^i}\right)\lrcorner\left(\deriv{m^\gamma}{x^k}dx^k\right)
=\frac{\hat{\Omega}^i}{\rho}\left(\deriv{m^\gamma}{x^i}\right), \label{eq:fr11}\\
\frac{d}{dt}\left(\frac{\Omega^\gamma}{\rho}\right)
=&\frac{d}{dt}\left(\frac{\hat{\Omega}^i}{\rho}\deriv{m^\gamma}{x^i}\right)=0, \label{eq:fr12}
\end{align}
are equivalent to Lie dragging ${\bf b}=(\hat{\Omega}^i/\rho)\partial/\partial x^i$ in (\ref{eq:fr10}). 

\end{proposition}

\begin{proof}
 The proof of (\ref{eq:fr6})-(\ref{eq:fr10}) are given in Webb et al. (2014a). 
To show that (\ref{eq:fr8}) is equivalent to (\ref{eq:5.62}), note that the sum:
\begin{align}
T=&({\bf e}_\alpha\times{\bf e}_\beta {\bf\cdot}\boldsymbol{\Omega}) dm^\alpha\otimes dm^\beta
\nonumber\\
=&\epsilon_{ijk}\left(\deriv{x^j}{m^\alpha} dm^\alpha\right)\otimes
\left(\deriv{x^k}{m^\beta} dm^\beta\right)\hat{\Omega^i}\nonumber\\
=&\left(\epsilon_{ijk}dx^j\otimes dx^k\right) \hat{\Omega}^i
=\sum_{j<k}\Omega_{jk} dx^j\wedge dx^k, \label{eq:fr13}
\end{align}
where $\Omega_{jk}=\epsilon_{ijk}\hat{\Omega}^i$ is the dual of $\hat{\Omega}^i$. Thus,
\begin{equation}
T=\boldsymbol{\Omega}{\bf\cdot} d{\bf S}
=\hat{\Omega}^x dy\wedge dz+\hat{\Omega}^y dz\wedge dx+\hat{\Omega}^z dx\wedge dy. \label{eq:fr14}
\end{equation}
Equation (\ref{eq:fr8}) is equivalent to $dT/dt=0$. Also equation (\ref{eq:5.62}) 
summed with weight factors 
$dm^\alpha\otimes dm^\beta$ is equivalent to $dT/dt=0$, which verifies that (\ref{eq:5.62})
(when summed with weight factors $dm^\alpha\otimes dm^\beta$) is equivalent to Lie dragging 
the vorticity 2-form $\boldsymbol{\Omega}{\bf\cdot}d{\bf S}$ described by (\ref{eq:fr8}).

To prove (\ref{eq:fr12})
write $b^i=\hat{\Omega}^i/\rho$ and note $d/dt=\partial/\partial t+{\bf u}{\bf\cdot}\nabla$. 
We obtain:
\begin{align}
\frac{d}{dt}\left(\frac{\hat{\Omega}^i}{\rho}\deriv{m^\gamma}{x^i}\right)= &\frac{d}{dt}
\left(b^i\deriv{m^\gamma}{x^i}\right)=\frac{db^i}{dt}\deriv{m^\gamma}{x^i}
+b^i \frac{d}{dt}\left(\deriv{m^\gamma}{x^i}\right)\nonumber\\
=&\left(\frac{db^i}{dt}-{\bf b}{\bf\cdot}\nabla{u^i}\right) \deriv{m^\gamma}{x^i}
+b^i \derv{x^i}\left(\frac{dm^\gamma}{dt}\right)
\equiv \left(\deriv{\bf b}{t}+[{\bf u},{\bf b}]\right)^i \deriv{m^\gamma}{x^i}=0. \label{eq:fr16}
\end{align}
In (\ref{eq:fr16}) we used the fact that $m^\gamma$ is a Lagrange label, satisfying $dm^\gamma/dt=0$.
\end{proof}

\subsection{Noether's second theorem}
In this section we derive 
(\ref{eq:fr8}) using 
Noether's second theorem,
in which $a=\rho d^3x$ is  Lie dragged 
 with the flow  
and by using  a gauge transformation or divergence symmetry of the action. We use the 
approach of Padhye (1996a,b), Cotter et al. (2007) 
and Webb et al. (2014b). An alternative approach is that of Hydon 
and Mansfield (2011) which uses Lagrange multipliers to incorporate constraints  
into the formulation (see e.g Webb and Mace (2015) for an application  
to potential vorticity related conservation laws in MHD). It is also possible to   
use the multi-symplectic form of Noether's second theorem in the analysis. 

Webb et al. (2014b) 
 showed that the action remains invariant under a fluid relabelling, divergence symmetry provided 
the invariance condition:
\begin{equation}
\nabla{\bf\cdot}\left(\rho \hat{V}^{\bf x}\right) \left(w+\Phi({\bf x})-\frac{1}{2}|{\bf u}|^2\right) 
+\rho T\hat{V}^{\bf x}{\bf\cdot}\nabla S +\rho {\bf u}{\bf\cdot}\left(\derv{t}+{\cal L}_{\bf u}\right)
\hat{V}^{\bf x}=-\nabla_\alpha\Lambda^\alpha, \label{eq:n1}
\end{equation}
is satisfied (see proposition 6.4 of Webb et al. (2014b) for MHD, but 
with ${\bf B}=0$). 
Here
\begin{equation}
\nabla_\alpha\Lambda^\alpha=\deriv{\Lambda^0}{t}+\deriv{\Lambda^i}{x^i}, \label{eq:n2}
\end{equation}
is  associated with a divergence transformation 
 in which:
\begin{equation}
L'=L+\epsilon D_\alpha\Lambda^\alpha, \label{eq:n3}
\end{equation}
is the transformation of the Lagrangian $L$. In (\ref{eq:n1}), ${\cal L}_{\bf u}$ is the Lie 
derivative with respect to ${\bf u}$, i. e.
\begin{equation}
\left(\derv{t}+{\cal L}_{\bf u}\right)\hat{V}^{\bf x}\equiv \deriv{\hat{V}^{\bf x}}{t}+[{\bf u},\hat{V}^{\bf x}]=0, 
\label{eq:n4}
\end{equation}
corresponds to Lie dragging of $\hat{V}^{\bf x}$ by  ${\bf u}$. 
We use the Lagrangian map in which the Eulerian position of the fluid element 
${\bf x}={\bf x}({\bf m},t)$ depends on the Lagrange labels ${\bf m}$ and the time $t$. The 
infinitesimal Lie transformations in (\ref{eq:n1}) have the form:
\begin{equation}
{\bf x}'={\bf x}+\epsilon V^{\bf x},\quad t'=t+\epsilon V^t,\quad {\bf m}'={\bf m}+\epsilon V^{\bf m}, 
\label{eq:n5}
\end{equation}
The canonical form of the Lie transformations (\ref{eq:n5}) are:
\begin{equation}
{\bf m}'={\bf m}, \quad t'=t, \quad {\bf x}'={\bf x}+\epsilon \hat{V}^{\bf x} \quad \hbox{where}\quad 
\hat{V}^{\bf x}=V^{\bf x}-V^{\alpha}D_\alpha {\bf x}, \label{eq:n6}
\end{equation}
and $V^0=V^t$ and $V^i\equiv V^{m^i}$. For fluid relabelling symmetries,  
the physical variables do not change, and in that case $V^t=V^{\bf x}=0$. 

The invariance condition (\ref{eq:n1}) 
can be satisfied if:
\begin{align}
&\nabla{\bf\cdot}\left(\rho \hat{V}^{\bf x}\right)=0, \label{eq:n7}\\
&\left(\derv{t}+{\cal L}_{\bf u}\right)\hat{V}^{\bf x}=0, \label{eq:n8}\\
&\rho T\hat{V}^{\bf x}{\bf\cdot}\nabla S=-\nabla_\alpha\Lambda^\alpha. \label{eq:n9}
\end{align}
Conditions (\ref{eq:n7}) and (\ref{eq:n9}) are satisfied if 
\begin{align}
\rho\hat{V}^{\bf x}=&\nabla\times \boldsymbol{\psi}:={\bf C}\quad\hbox{or}
\quad \hat{V}^{\bf x}=\frac{\nabla\times\boldsymbol{\psi}}{\rho}\equiv \frac{\bf C}{\rho}, 
\label{eq:n10}\\
\Lambda^0=&r {\bf C}{\bf\cdot}\nabla S,\quad \Lambda^i=\left(r {\bf C}{\bf\cdot}\nabla S\right) u^i,
\quad 1\leq i\leq 3. \label{eq:n11}
\end{align}
Condition (\ref{eq:n8}) requires that $\hat{V}^{\bf x}=\nabla\times \boldsymbol{\psi}/\rho$ 
is Lie dragged 
by the flow. We also require that $a=\rho d^3x$ is  Lie dragged by ${\bf u}$.
The quantity:
\begin{equation}
\hat{V}^{\bf x}\lrcorner\left(\rho d^3x\right)=\left(\hat{V}^{x^i}\derv{x^i}\right)\lrcorner\left(\rho d^3x\right)
=\rho \hat{V}^{\bf x}{\bf\cdot}d{\bf S}=\nabla\times\boldsymbol{\psi}{\bf\cdot}d{\bf S}
=d\left(\boldsymbol{\psi}{\bf\cdot}d{\bf x}\right), \label{eq:n12}
\end{equation}
defines a Lie dragged invariant 2-form $\delta={\bf C}{\bf\cdot}d{\bf S}
=\nabla\times\boldsymbol{\psi}{\bf\cdot}d{\bf S}$. 
Also note that $\delta=d\alpha$ where $\alpha=\boldsymbol{\psi}{\bf\cdot}d{\bf x}$.
Using the algebra of exterior differential forms and Cartan's magic formula, (e.g. Webb et al. (2014a)),
we find:
\begin{align}
\frac{d}{dt}\left({\bf C\cdot}d{\bf S}\right)=&\left(\deriv{\bf C}{t}
-\nabla\times({\bf u}\times {\bf C})+{\bf u}\left(\nabla{\bf\cdot C}\right)\right)
{\bf\cdot}d{\bf S}=0, \label{eq:n13}\\
\frac{d}{dt}\left(\boldsymbol{\psi}{\bf\cdot}d{\bf x}\right)
=&\left(\deriv{\boldsymbol{\psi}}{t}
-{\bf u}\times\left(\nabla\times\boldsymbol{\psi}\right)
+ \nabla\left({\bf u\cdot}\boldsymbol{\psi}\right)\right){\bf\cdot}d{\bf x}=0. 
\label{eq:n14}
\end{align}

To verify that the balance equation (\ref{eq:n9}) for invariance of the action is satisfied 
by the gauge potential solutions (\ref{eq:n11}) for the $\Lambda^\alpha$ ($0\leq \alpha\leq 3$), 
we note that $-\nabla_\alpha\Lambda^\alpha$ can be reduced to:
\begin{align}
-\nabla_\alpha\Lambda^\alpha=&-\left({\bf C\cdot}\nabla S\right) \frac{dr}{dt}
-r\left(\derv{t}\left({\bf C\cdot}\nabla S\right)
+\nabla{\bf\cdot}\left({\bf C}{\bf\cdot}\nabla S\ {\bf u}\right)\right)\nonumber\\
&\equiv \left({\bf C\cdot}\nabla S\right) T
-r\rho\frac{d}{dt}\left(\frac{{\bf C\cdot}\nabla S}{\rho}\right)\nonumber\\ 
&=\left({\bf C\cdot}\nabla S\right) T=T\left(\rho \hat{V}^{\bf x}{\bf\cdot}\nabla S\right), 
\label{eq:n15}
\end{align}
which verifies that the condition (\ref{eq:n9}) for a divergence symmetry of the action is 
satisfied by the solution ansatz (\ref{eq:n11}). Note that ${\bf C}{\bf\cdot}\nabla S/\rho$
is the inner product of the Lie dragged vector field $\hat{V}^{\bf x}{\bf\cdot}\nabla$ and the 
Lie dragged 1-form $\nabla S{\bf\cdot}d{\bf x}$  and  hence is an advected scalar invariant, 
i.e. $d/dt({\bf C}{\bf\cdot}\nabla S/\rho)=0$.  

Webb et al. (2014b) used an Eulerian, 
Euler-Poincar\'e variational approach (see also Cotter et al. (2007)). They obtained:
\begin{equation}
\delta J=\int\int \hat{V}^{\bf x}{\bf\cdot} {\bf E}(\ell)\ d^3x dt+\int\int \left(\deriv{D}{t}+\nabla{\bf\cdot}{\bf F}\right)\ d^3xdt, \label{eq:n16}
\end{equation}
for the variation of the action $J=\int\int \ell\ d^3xdt$, where ${\bf E}(\ell)=0$ are the 
Euler Lagrange equations for the system (i.e.the Eulerian momentum equations). They 
used the Lagrangian map in which ${\bf x}={\bf x}({\bf m},t)$ specifies the Eulerian fluid element 
position in terms of the Lagrangian mass coordinates ${\bf m}$ and the time $t$. 
The variation  $\delta J$
of the action under Lie and divergence transformations of the action, gave:
\begin{align}
D=&\rho \hat{V}^{\bf x}{\bf\cdot}{\bf u}+\Lambda^0, \label{eq:n17}\\
{\bf F}=&\rho \hat{V}^{\bf x}{\bf\cdot}
\biggl[{\bf u}\otimes {\bf u}+\left(w+\Phi
-\frac{1}{2} |{\bf u}|^2\right) {\sf I}\biggr]+\boldsymbol{\Lambda}. \label{eq:n18}
\end{align}
 
Dropping the divergence term $\nabla{\bf\cdot F}$ in (\ref{eq:n16}) since it gives rise to 
a surface integral over the spatial boundary (assumed to be at infinity)  which is assumed 
to vanish, and using (\ref{eq:n10}) for $\hat{V}^{\bf x}$ we obtain:
\begin{equation}
\delta J=\int\int\frac{\nabla\times\boldsymbol{\psi}}{\rho}{\bf\cdot}{\bf E}(\ell) d^3x dt
+\int\int\deriv{D}{t}d^3x dt. \label{eq:n19}
\end{equation}
From (\ref{eq:n11}) and (\ref{eq:n17}) we obtain:
\begin{equation}
D=\rho\hat{V}^{\bf x}{\bf\cdot}\left({\bf u}+r\nabla S\right)\equiv 
\left(\nabla\times\boldsymbol{\psi}\right){\bf\cdot}{\bf v}, \label{eq:n20}
\end{equation}
where
\begin{equation}
{\bf v}={\bf u}+r\nabla S\quad\hbox{and}\quad \boldsymbol{\Omega}=\nabla\times {\bf v}
=\boldsymbol{\omega}+\nabla r\times\nabla S, \label{eq:n21}
\end{equation}
and $\boldsymbol{\omega}=\nabla\times{\bf u}$ is the fluid vorticity. 

Using (\ref{eq:n20}) and (\ref{eq:n21}), the second term $\delta J_2$ in (\ref{eq:n19}) reduces to:
\begin{equation}
\delta J_2=\int\int \derv{t}\left[\nabla{\bf\cdot}\left(\boldsymbol{\psi}\times{\bf v}\right)
+\boldsymbol{\Omega}{\bf\cdot}\boldsymbol{\psi}\right]\ d^3x dt
\equiv \int\int \left(\boldsymbol{\Omega}_t {\bf\cdot} \boldsymbol{\psi}
+\boldsymbol{\Omega}{\bf\cdot}\boldsymbol{\psi}_t\right)\ d^3x dt, \label{eq:n22}
\end{equation}
where we dropped the surface divergence term in the last step. 
Using the partial differential equation (\ref{eq:n14}) for $\boldsymbol{\psi}$ 
to eliminate $\boldsymbol{\psi}_t$ in (\ref{eq:n22})  gives:
\begin{align}
\delta J_2=&\int\int \boldsymbol{\psi}{\bf\cdot}\boldsymbol{\Omega}_t+\boldsymbol{\Omega}
{\bf\cdot}\left[{\bf u}\times(\nabla\times\boldsymbol{\psi})
-\nabla({\bf u}{\bf\cdot}\boldsymbol{\psi}\right]\ d^3x dt\nonumber\\
=&\int\int \biggl\{\boldsymbol{\psi}{\bf\cdot}
\left[\boldsymbol{\Omega}_t -\nabla\times({\bf u}\times\boldsymbol{\Omega}) 
+{\bf u}(\nabla{\bf\cdot}\boldsymbol{\Omega})\right]\nonumber\\
&+\nabla{\bf\cdot}
\left[\boldsymbol{\psi}\times(\boldsymbol{\Omega}\times{\bf u})
-({\bf u}{\bf\cdot}\boldsymbol{\psi})\boldsymbol{\Omega}\right]
\biggr\}d^3x dt\nonumber\\
\equiv&\int\int \boldsymbol{\psi}{\bf\cdot}
\left[\boldsymbol{\Omega}_t -\nabla\times({\bf u}\times\boldsymbol{\Omega})
+{\bf u}(\nabla{\bf\cdot}\boldsymbol{\Omega})\right]\ d^3xdt, \label{eq:n23}
\end{align}
where we dropped the pure divergence surface term in the last step. 

Using integration by parts, the first integral $\delta J_1$ in (\ref{eq:n19}) 
involving ${\bf E}(\ell)$ reduces to:
\begin{align}
\delta J_1=&\int\int \nabla{\bf\cdot}\left(\frac{\boldsymbol{\psi}\times{\bf E}(\ell)}{\rho}\right)
+\boldsymbol{\psi}{\bf\cdot}\nabla\times\left(\frac{{\bf E}(\ell)}{\rho}\right)\ d^3x dt
\nonumber\\
\equiv &\int\int \boldsymbol{\psi}{\bf\cdot}\nabla\times
\left(\frac{{\bf E}(\ell)}{\rho}\right)\ d^3x dt,
\label{eq:n24}
\end{align}
where the pure divergence term has been dropped. Adding $\delta J_1$ and $\delta J_2$ we get 
the total variation $\delta J=\delta J_1+\delta J_2$ as:
\begin{equation}
\delta J=\int\int \boldsymbol{\psi}{\bf\cdot}\biggl\{\nabla\times
\left(\frac{{\bf E}(\ell)}{\rho}\right)
+\boldsymbol{\Omega}_t -\nabla\times({\bf u}\times\boldsymbol{\Omega})
+{\bf u}(\nabla{\bf\cdot}\boldsymbol{\Omega})\biggr\}\ d^3xdt,  \label{eq:n25}
\end{equation}
and hence $\delta J=0$ if $\boldsymbol{\Omega}$ satisfies the generalized Bianchi identity:
\begin{equation}
\nabla\times\left(\frac{{\bf E}(\ell)}{\rho}\right)
+\boldsymbol{\Omega}_t -\nabla\times({\bf u}\times\boldsymbol{\Omega})
+{\bf u}(\nabla{\bf\cdot}\boldsymbol{\Omega})=0. \label{eq:n26}
\end{equation}
Thus, if the Euler Lagrange equations ${\bf E}(\ell)=0$ are satisfied, (\ref{eq:n26}) gives the 
vorticity flux conservation law (\ref{eq:fr8}).

 Webb et al. (2014b) used Noether's second theorem, with 
\begin{align}
&\hat{V}^{\bf x}=\frac{\boldsymbol{\Omega}}{\rho},\quad V^t=V^{\bf x}=0,\quad  
V^{{\bf x}_0}=-\frac{\nabla_0\lambda\times\nabla_0\mu}{\rho_0}, \nonumber\\
&\Lambda^0=r \boldsymbol{\Omega}{\bf\cdot}\nabla S,\quad \Lambda^i= ru^i 
(\boldsymbol{\Omega}{\bf\cdot}\nabla S), \label{eq:n27}
\end{align}
to obtain the generalized helicity conservation law:
\begin{equation}
\derv{t}\left[\boldsymbol{\Omega}{\bf\cdot}({\bf u}+r\nabla S)\right] 
+\nabla{\bf\cdot}\left\{{\bf u}\left[\boldsymbol{\Omega}{\bf\cdot}({\bf u}+r\nabla S)\right]
+\boldsymbol{\Omega}(w-\frac{1}{2} u^2)\right\}=0. \label{eq:n28}
\end{equation}
 This conservation law can also be 
derived by using the multi-symplectic form of Noether's theorem (e.g. Hydon (2005), 
Hydon and Mansfield (2011), Webb et al. (2014c)). The above result is a consequence of 
a fluid relabelling symmetry ($V^{{\bf x}_0}\neq 0$) coupled with a 
gauge transformation (i.e. $\Lambda^0\neq 0$ and 
$\Lambda^i\neq 0$). Since $r$ is a nonlocal variable, (\ref{eq:n28}) 
is a nonlocal conservation law.

\section{Differential forms approach}
The results of this section on differential forms representation of the multi-symplectic 
system (\ref{eq:5.1}) have been derived  by Webb (2015) for the case of 1D 
Lagrangian gas dynamics. The present formulation applies to the more general case 
where there are $n$ Lagrange mass coordinate labels $m^i$, ($1\leq i\leq n$).   
We omit the $\lambda^k$ and $\mu^k$ terms in the Lagrangian in the analysis below, 
since their evolution is decoupled from the rest of the equations (i.e. we effectively 
set $\mu^k=\lambda^k=0$ in the analysis).
\subsection{Variational principles}
\begin{proposition}
Consider the variational functional:
\begin{equation}
J=\int\ \psi^*(\Theta)\equiv \int_M \tilde{\ell}_m dV, \label{eq:df1}
\end{equation}
where $\psi^*(\Theta)$ is the pullback of the differential form $\Theta$: 
\begin{align}
\Theta=&\omega^\alpha\wedge d\tilde{m}_\alpha-H dV, \label{eq:df2}\\
dV=&dt\wedge dm^1\wedge dm^2\wedge \ldots\wedge dm^n,\quad d\tilde{m}_\mu=\partial_\mu\contr dV, 
\label{eq:df3}
\end{align}
In (\ref{eq:df1})-(\ref{eq:df3}) $\tilde{\ell}_m$ is 
the multi-symplectic Lagrangian (\ref{eq:el10}). 
The stationary point conditions for the action (\ref{eq:df1}): $\delta J/\delta z^i=0$ 
give the multi-symplectic system (\ref{eq:5.1}). In particular:  
\begin{align}
h_m=&\frac{1}{2} u^2+e(\tau,S)
+\Phi({\bf x})+\pi_{ik} x_{ik}, \nonumber\\
\tilde{\ell}_m=&r \frac{dS}{dt}+u^k\deriv{x^k}{t}
-\left(\frac{1}{2} u^2+e(\tau,S)
+\Phi({\bf x})\right), \label{eq:df4}
\end{align}
are the multi-symplectic Hamiltonian density (\ref{eq:5.1b}) 
and the multi-symplectic Lagrangian density (\ref{eq:el10}). In 
(\ref{eq:df4}) $\tau=J=1/\rho$ where $J=\det(x_{ij})$. 
\end{proposition}

\begin{proof}

To prove (\ref{eq:df1}), note that:
\begin{align}
&\psi^*\left(\omega^\mu\wedge d\tilde{m}_\mu\right)
=\left(L^\mu_j dz^j\right)\wedge d\tilde{m}_\mu\nonumber\\
&= L^\mu_j \deriv{z^j}{m^s} dm^s\wedge \left[(-1)^\mu dm^0\wedge\ldots \wedge dm^{\mu-1}
\wedge dm^{\mu+1}\ldots\wedge dm^n\right]\nonumber\\
&=L^\mu_j \deriv{z^j}{m^s} (-1)^{2\mu} \delta^s_\mu dV
= L^\mu_j \deriv{z^j}{m^\mu} dV. 
\label{eq:df5}
\end{align}
Thus,
\begin{equation}
\psi^*(\Theta)=\left(L^\mu_j \deriv{z^j}{m^\mu}-h_m(z)\right) dV\equiv \tilde{\ell}_m dV, 
\label{eq:df6}
\end{equation}
where in the last step, the standard Legendre transformation 
between the multi-symplectic Hamiltonian $h_m$ 
and Lagrangian $\tilde{\ell}_m$  has been used (e.g. Hydon (2005)). Note that 
$\partial\tilde{\ell}_m/\partial (z^j_{,\mu})=\pi^\mu_j=L^\mu_j$ is the canonical 
multi-momentum corresponding to $z^j_{,\mu}$. Calculating the variational derivative 
$\delta J/\delta z^i$ of the action 
(\ref{eq:df1}) gives the Euler Lagrange equations $\delta J/\delta z^i=0$: 
\begin{equation}
\frac{\delta J}{\delta z^i}=\deriv{\tilde{\ell}_m}{z^i}
-\derv{m^\mu}\left(\deriv{\tilde{\ell}_m}{z^i_{,\mu}}\right)=
\left({\sf K}^\mu_{ij}\deriv{z^j}{m^\mu}-\deriv{h_m}{z^i}\right)=0, \label{eq:df7}
\end{equation}
which is the multisymplectic system (\ref{eq:5.1}). This completes the proof.
\end{proof}

\begin{proposition}
Consider the variational functional 
\begin{equation}
K[\Omega]=\int_M \Omega, \label{eq:vbun1}
\end{equation}
where $M$ is a region with boundary $\partial M$ in the fiber bundle space in which the $z^i$ are 
regarded as independent of the base manifold coordinates $m^\alpha=(t,m^1,\ldots m^n)$. The form:
\begin{equation}
\Omega=d\Theta=d\omega^\alpha\wedge d\tilde{m}_\alpha-dH\wedge dV
\quad\hbox{where}\quad dV=dt\wedge dm^1\wedge\ldots\wedge dm^n, \label{eq:vbun2}
\end{equation}
is  the Cartan-Poincar\'e form for the system (\ref{eq:5.1}).  
Variations of the functional (\ref{eq:vbun1}) 
are described by the Lie derivative:
\begin{equation}
{\cal L}_{\bf V}=\frac{d}{d\epsilon}=V^i\derv{z^i}, \label{eq:vbun3}
\end{equation}
where the base manifold variables $m^\alpha$ are fixed, and ${\bf V}$ 
is an arbitrary but smooth vector field. The variations of $K[\Omega]$ 
are described by:
\begin{equation}
\delta K[\Omega]=\int_M {\cal L}_{\bf V}\left(\Omega\right). \label{eq:vbun4}
\end{equation}
The variations (\ref{eq:vbun4}) can be reduced to the form:
\begin{equation}
\delta K[\Omega]=\int_{\partial M} V^p\beta_p, \label{eq:vbun5}
\end{equation}
where the forms 
$\{\beta_p:\ 1\leq p\leq N\}$ ($N$ is the number of $z^i$ variables) 
are given by:
\begin{equation}
\beta_p={\sf K}^\alpha_{pj} dz^j\wedge d\tilde{m}_\alpha-\deriv{H}{z^p} dV, \label{eq:vbun6}
\end{equation}
and $\partial M$ is the boundary of the region $M$ in the ${\bf z}$-space. 
The equations $\beta_p=0$ ($1\leq p\leq N)$, provide a basis of Cartan 
forms for the multi-symplectic system (\ref{eq:5.1}). The pullback of the
forms  $\beta_p$ to the base manifold  gives the equations:
\begin{equation}
\tilde{\beta}_p=\left({\sf K}^\alpha_{pj}\deriv{z^j}{m^\alpha}-\deriv{H}{z^p}\right) dV. 
\label{eq:vbun7}
\end{equation}
The sectioned forms equations $\tilde{\beta}_p=0$ give  the multi-symplectic 
partial differential equation system (\ref{eq:5.1}).  
\end{proposition}
\begin{proof}
This proposition was proved for the case of 1D Lagrangian gas dynamics (the $n=1$ case) 
by Webb (2015). Essentially the same proof applies for $n>1$.
An essential formula in the proof is
Cartan's magic formula:
\begin{equation}
{\cal L}_{\bf V}\Omega={\bf V}\lrcorner d\Omega+d\left({\bf V}\lrcorner \Omega\right)
\equiv d\left({\bf V}\lrcorner \Omega\right), \label{eq:vbun8}
\end{equation}
where we use $\Omega=d\Theta$ and $d\Omega=dd\Theta=0$. The detailed proof is given by Webb (2015). 

\end{proof}
\begin{remark}
Webb (2015) showed that for 1D gas dynamics, the $\beta_i$ form a closed ideal 
of forms representing the partial differential system (\ref{eq:5.1}).
\end{remark}

\subsection{The differential forms $\beta_p$}
Below we give explicit formulae for the differential forms in (\ref{eq:vbun6}), for 
the case of $n$ Lagrange labels $m^k$ where $1\leq k\leq n$. We use the notation $m^0=t$ 
for the time variable. The case $n=1$ has been considered in detail by Webb (2015) 
who considered the case of 1D Lagrangian gas dynamics using a similar analysis to the 
present paper. The number of equations for the forms $\beta_p$ is $N=2 n^2+2n+2$. 
Our formulae hold for both the case of 2D gas dynamics ($n=2$) and also for the case of 3D gas 
dynamics ($n=3$). 
 The differential forms $\{\beta_p:\ 1\leq p\leq N\}$ in (\ref{eq:vbun6}) describes 
the multi-symplectic system (\ref{eq:5.1}), i.e. the pullback of the forms $\beta_p$, 
i.e. $\tilde{\beta}_p=0$ gives the multi-symplectic system:
\begin{equation}
{\sf K}^\alpha_{pj}\deriv{z^j}{m^\alpha}-\deriv{H}{z^p}=0. \label{eq:he1}
\end{equation}
Below we list explicitly formulae for the $\beta_p$.  We find for $1\leq i\leq n$:
\begin{align}
\beta_i=&-\left(du^i\wedge d\tilde{m}_0+d\pi_{ij}\wedge d\tilde{m}_j
+\deriv{\Phi}{x^i}dV\right), \nonumber\\
\tilde{\beta}_i=&-\left(\frac{du^i}{dt}+\deriv{\pi_{ij}}{m^j}+\deriv{\Phi}{x^i}\right) dV. \label{eq:he2}
\end{align}
Thus the equations $\tilde{\beta}_i=0$ ($1\leq i\leq n$) give the momentum equations (\ref{eq:5.15}) 
for the gas dynamic equations. 

For $n+1\leq i\leq 2n$, and $i=j+n$ ($1\leq j\leq n$) (\ref{eq:vbun6}) gives:
\begin{align}
\beta_i=&\beta_{j+n}=\beta^{u^j}=dx^j\wedge d\tilde{m}_0-u^j dV\equiv 
(dx^j-u^j dt)\wedge d\tilde{m}_0, 
\nonumber\\
\tilde{\beta}^{u^j}
=&\left(\deriv{x^j({\bf m},t)}{t}-u^j\right) dV. \label{eq:he3}
\end{align}
In the first equation (\ref{eq:he3}) we used $t=m^0$ and $dm^0\wedge d\tilde{m}_0=dV$. 
Note that the equations $\tilde{\beta}_{j+n}=0$ are equivalent to (\ref{eq:5.16}). 

For $2n+1\leq i\leq n^2+2n$, we write $i=nk+j+n$ where $1\leq k\leq n$ and $1\leq j\leq n$ and obtain:
\begin{align}
\beta_i=&\beta_{nk+j+n}
=\beta_{jk}\quad\hbox{where}\nonumber\\
\beta_{jk}=&dx^j\wedge d\tilde{m}_k-x_{jk} dV\equiv (dx^j-x_{jk}dm^k)\wedge d\tilde{m}_k, \label{eq:he4} 
\end{align}
The pullback equations $\tilde{\beta}_{jk}=0$ give equations (\ref{eq:5.18a}). Note in (\ref{eq:he4}) 
that $k$ is fixed and there is no sum over $k$.

For $i=n^2+2n+1$ and $i=n^2+2n+2$ we obtain:
\begin{align}
\beta^r=&\beta_{(n^2+2n+1)}=-\left(dr\wedge d\tilde{m}_0+TdV\right)\equiv -(dr+Tdt)\wedge d\tilde{m}_0,
\nonumber\\ 
\beta^S=&\beta_{n^2+2n+2}=\frac{dS}{dt} dV\equiv dS\wedge d\tilde{m}_0,\label{eq:he5} 
\end{align}
The pullback equations $\tilde{\beta}_{r}=0$ and $\tilde{\beta}_{S}=0$ are 
given in (\ref{eq:5.18}). 

For $n^2+2n+3\leq i\leq 2n^2+2n+2$, we write $i=nk+j+n^2+n+2$ where 
$1\leq k\leq n$ and $1\leq j\leq n$ and obtain:
\begin{equation}
\beta_i=\beta_{(nk+j+n^2+n+2)}=\mu_{kj}\quad \hbox{where}\quad \mu_{kj}=-\left(\pi_{kj}-p A_{kj}\right) dV. 
\label{eq:he6}
\end{equation}
The pullback equations $\tilde{\mu}_{kj}=0$ give (\ref{eq:5.18a}). 
Equations (\ref{eq:he1})-(\ref{eq:he6}) show the differential 
forms $\{\beta_p\}$ are related to the pullback equations $\tilde{\beta}_p=0$
which are equivalent to the multi-symplectic system (\ref{eq:5.1}).  

\subsubsection{The ideal of differential forms}
The set of forms:
\begin{equation}
\mathscr{I}=\left\{\beta_j,\beta^{u^j},\beta_{kj},\beta^r,\beta^S\right\}
\quad\hbox{where}\quad 1\leq j,k\leq n, \label{eq:id1}
\end{equation}
defined by the equations:
\begin{align}
\beta_j=&-\left(du^j\wedge d\tilde{m}_0+d\pi_{jk}\wedge d\tilde{m}_k
+\deriv{\Phi}{x^j}dV\right), \nonumber\\
\beta^{u^j}=&dx^j\wedge d\tilde{m}_0-u^j dV\equiv
(dx^j-u^j dt)\wedge d\tilde{m}_0, \nonumber\\
\beta_{jk}=&dx^j\wedge d\tilde{m}_k-x_{jk} dV
\equiv \left(dx^j-x_{jk}dm^k\right)\wedge d\tilde{m}_k, \nonumber\\
\beta^S=&dS\wedge d\tilde{m}_0,\quad
\beta^r=-\left(dr\wedge d\tilde{m}_0+TdV\right)\equiv -(dr+Tdt)\wedge d\tilde{m}_0, \label{eq:id2}
\end{align}
 is a closed ideal of forms representing the 
Lagrangian gas dynamic system (\ref{eq:5.1}).  
 This result covers both the cases $n=2$ and $n=3$ where $n$ is the number 
of Lagrangian mass coordinates $m^j$.   The above basis of forms (\ref{eq:id1})
can be used to represent the Lagrangian fluid dynamic equations in 
Cartan's geometric theory of partial differential equations (e.g. Harrison and Estabrook, 1971). 
 \begin{proposition}
The set of exterior differential forms defined in (\ref{eq:id1})-(\ref{eq:id2}) 
is a closed ideal of forms in the sense of Cartan's geometric theory of partial differential
equations (e.g. Harrison and Estabrook (1971)). This requires that the exterior derivatives of the 
set of forms (\ref{eq:id1})-(\ref{eq:id2})  may be expressed as a linear combination of the basis forms 
involving the wedge product. In particular, 
\begin{align}
d\beta_j=&(-1)^{n+1} \Phi_{,js} \left(\beta^{u^s}\wedge dt\right), \label{eq:id3}\\
d\beta^{u^j}=&\left(-1\right)^n \beta_j\wedge dt, \label{eq:id4}\\
d\beta^S=&0, \label{eq:id5}\\
d\beta^r=&-\left(\tilde{w}_{Sp} dp\wedge dV+(-1)^n \tilde{w}_{SS}\beta^S\wedge dt\right), \label{eq:id6}\\
dp\wedge dV=&\frac{(-1)^{n+1}}{D}\left\{\tilde{w}_{pS} \beta^S\wedge dt
+\frac{x_{ij}}{(n-1)p} \beta_i \wedge dm^j\right\},\label{eq:id7}\\
D=&\left(
\tilde{w}_{pp}+\frac{n\tau}{(n-1)p}\right),  \label{eq:id8}\\
d\beta_{jk}=&\frac{-1}{p\tau}
\biggl\{(-1)^n x_{js} x_{ik} (\beta_i\wedge dm^s)\nonumber\\
&+x_{jk}\left[\left(p \tilde{w}_{pp}+\tau\right)dp\wedge dV
+(-1)^n p\tilde{w}_{pS} \beta^S\wedge dt\right]\biggr\}. \label{eq:id9}
\end{align}
In (\ref{eq:id3}) $\Phi_{,is}=\partial^2\Phi/\partial x^i\partial x^s$. Note that the exterior derivatives 
of the set of forms in (\ref{eq:id1})-(\ref{eq:id2}) is closed. This result depends on the expansion for 
$dp\wedge dV$ in (\ref{eq:id7}).
\end{proposition}

\begin{proof}
From (\ref{eq:id2})
\begin{align}
d\beta_j=&-\frac{\partial^2 \Phi}{\partial x^j\partial x^s} dx^s\wedge dV, \label{eq:id11}\\
\beta^{u^s}\wedge dm^0=&dx^s\wedge d\tilde{m}_0\wedge dm^0=(-1)^ndx^s\wedge dV. \label{eq:id12}
\end{align}
Use of (\ref{eq:id12}) in (\ref{eq:id11}) gives the result (\ref{eq:id3}) for $d\beta_j$.
Using (\ref{eq:id2}) we obtain:
\begin{align}
\beta_j\wedge dm^0=&-du^j\wedge d\tilde{m}_0\wedge dm^0=(-1)^{n+1} du^j\wedge dV, \nonumber\\
d\beta^{u^j}=&-du^j\wedge dV=(-1)^n\beta_j\wedge dt, \label{eq:id13}
\end{align}
which establishes (\ref{eq:id4}) (note $m^0=t$). 

By noting $\beta^S=d(S\wedge d\tilde{m}_0)$ in (\ref{eq:id2}) we obtain $d\beta^S=dd(S d\tilde{m}_0)=0$
which establishes (\ref{eq:id5}). 

To derive (\ref{eq:id6}) and (\ref{eq:id7}) first note that 
\begin{equation}
\tilde{w}_p=\tau\equiv \det(x_{ij})=J, \quad \tilde{w}_S=T. \label{eq:id14}
\end{equation}
Also from (\ref{eq:id2}) and (\ref{eq:id14}), we obtain:
\begin{align}
d\beta^r=&-dT\wedge dV=-d\tilde{w}_S\wedge dV
=-\left(\tilde{w}_{Sp} dp+\tilde{w}_{SS}dS\right)\wedge dV\nonumber\\
=&-\left(\tilde{w}_{Sp}dp\wedge dV+(-1)^n \tilde{w}_{SS} \beta^S\wedge dt\right), \label{eq:id15}
\end{align}
which establishes (\ref{eq:id6}) for $d\beta^r$. 

To derive (\ref{eq:id7}), note from (\ref{eq:id14}) that
\begin{equation}
d\tau=\tilde{w}_{pp}dp+\tilde{w}_{pS} dS. \label{eq:id16}
\end{equation}
Also since $\tau=J$ we obtain:
\begin{equation}
d\tau=\deriv{J}{x_{ij}}d x_{ij}=A_{ij} dx_{ij}=d(A_{ij}x_{ij})-x_{ij}dA_{ij}\equiv d (n\tau)-x_{ij}dA_{ij}. 
\label{eq:id17}
\end{equation}
From (\ref{eq:id16}) and (\ref{eq:id17}) and noting $A_{ij}=\pi_{ij}/p$, we obtain:
\begin{align}
d\tau\wedge dV=&\left(\tilde{w}_{pp}dp+\tilde{w}_{pS}dS\right)\wedge dV
=\tilde{w}_{pp}dp\wedge dV+(-1)^n \tilde{w}_{pS}\beta^S\wedge dt\nonumber\\
\equiv& \frac{1}{(n-1)p} \left(x_{ij} d\pi_{ij}\wedge dV-nJ dp\wedge dV\right). \label{eq:id18}
\end{align}
The last line in (\ref{eq:id18}) follows by using the alternative expression (\ref{eq:id17}) for $d\tau$. 
Solving (\ref{eq:id18}) for $dp\wedge dV$ gives:
\begin{equation}
dp\wedge dV=\frac{1}{D}\left(-(-1)^n \tilde{w}_{pS}\beta^S\wedge dt+\frac{x_{ij}d\pi_{ij}\wedge dV}{(n-1)p}\right),
\label{eq:id19}
\end{equation}
which can be reduced to the form (\ref{eq:id7}) where $D$ is defined in (\ref{eq:id8}). 

To derive (\ref{eq:id9}) we use (\ref{eq:id2}) to obtain:
\begin{equation}
d\beta_{jk}=-d x_{jk}\wedge dV. \label{eq:id20}
\end{equation}
Using the relation:
\begin{equation}
\pi_{ij} x_{ik}=pA_{ij} x_{ik}=p\tau\delta_{jk}, \label{eq:id21}
\end{equation}
we obtain:
\begin{equation}
d x_{jk}=\frac{1}{p\tau}\left[-x_{js}x_{ik}d\pi_{is}+x_{jk}\left(\tau dp+p d\tau\right)\right] 
\label{eq:id22}
\end{equation}
and
\begin{align}
d x_{jk}\wedge dV=&\frac{1}{p\tau}\left\{-x_{js}x_{ik} d\pi_{is}\wedge dV
+x_{jk}\left[\tau dp\wedge dV+pd\tau\wedge dV\right]\right\}, \nonumber\\
\equiv&\frac{1}{p\tau}\left\{-x_{js}x_{ik} d\pi_{is}\wedge dV
+x_{jk}\left[(p\tilde{w}_{pp}+\tau)dp\wedge dV+(-1)^n p \tilde{w}_{pS} \beta^S\wedge dt\right]\right\}. 
\label{eq:id23}
\end{align}
Then noting that:
\begin{equation}
d\pi_{is}\wedge dV=(-1)^{n+1} \beta_i\wedge dm^s, \label{eq:id24}
\end{equation}
(\ref{eq:id20}) and (\ref{eq:id23}) give (\ref{eq:id9}). This completes the proof.
\end{proof}

\section{Concluding Remarks}
In this paper, we obtained multi-symplectic equations for compressible, Lagrangian fluid dynamics, 
similar to the work of Bridges et al. (2005), except that our analysis applies for the 
case of non-barotropic fluids, in which the energy density $\varepsilon(\rho,S)$ per unit volume,
and the pressure $p=p(\rho,S)$ depend on both the density $\rho$ and entropy $S$. 
This case is different than the examples studied by Bridges et al. (2005), where 
explicit reference to the entropy is not mentioned.  We include 
an external gravitational potential $\Phi({\bf x})$  
 in our analysis. Bridges et al. (2005) considered several different fluid models, 
including the incompressible, shallow water equations in two Cartesian space dimensions, 
 including the effects of the Coriolis force due to  a rotating reference frame.   
They also studied the 3D incompressible fluid equations. 

The model (Section~2) and Lagrangian action principle (Section~3) for the equations were established.
 The Hamiltonian approach in which time $t$ is the evolution 
variable for the system was developed to give the canonical Hamiltonian Poisson bracket 
for the system. Section~4 gives the de Donder Weyl, multi-momentum approach to the 
fluid equations, which leads to the multi-symplectic form of the equations 
(Section~5). The pullback  
 and symplecticity conservation laws for the gas dynamic equations 
were established (Section~5) using the approach of Hydon (2005).
 The symplecticity 
 laws correspond to the  phase space conservation 
laws for multi-symplectic systems. They are 
 compatibility conditions for the pullback conservation laws. One class of symplecticity  
laws corresponds to setting the derivatives of the co-moving energy conservation law
with respect to the mass coordinates $m^i$ ($1\leq i\leq 3$) equal to zero. The second 
class of symplecticity laws correspond to vorticity conservation laws. 

Both the pullback conservation laws 
and the vorticity-symplecticity conservation laws are nonlocal as they depend on the 
nonlocal Clebsch variable $r$, where $dr/dt=-T$ where $T$ is the temperature of the gas.
These results are significant in the description of vorticity evolution in atmospheric 
dynamics where baroclinicity (nonalignment of density and pressure contour levels) 
is a source of vorticity in tornadoes and other vorticity phenomena such as Rossby waves
(e.g. Pedlosky 1979, Rhines 2003, Vallis 2006, Gao et al. (2012), Yang et al. (2014)). In atmospheric dynamics, a reference frame 
rotating with the Earth (or planetary body or star) is used, in which case,  
the Coriolis force, the centrifugal force and Darwin force 
must be taken into account (e.g. Holm 2008). The extension of the present analysis 
to Lagrangian magnetohydrodynamics will be investigated in a separate paper.  

By using Lie dragging techniques (e.g. Tur and Yanovsky (1993), Webb et al. (2014a)) 
the vorticity symplecticity laws gives the  
 2-form vorticity flux conservation law (Section~6).
The 2-form vorticity  
law applies to the component of the fluid vorticity that is independent of the entropy
gradients.  
The vorticity 2-form law also arises from Noether's second theorem, 
and the mass conservation fluid relabelling symmetry and a divergence 
symmetry of the action.   
Differential form representation of the equations using 
the Cartan-Poincar\'e $n+2$ form were obtained, where $n$ is the number of independent 
Lagrangian mass coordinates (Section~7). This approach also leads to the 
action principle for the system. 
 The Cartan-Poincar\'e form  gives a set of differential 
forms representing  the partial differential equation system (e.g. Harrison and Estabrook (1971)),
which may be used to investigate the Lie symmetries of the equations.

The main point established by the paper, is that vorticity and symplecticity 
are closely related concepts (see also Bridges et al. (2005)).  
\appendix
\section*{Appendix A}
\setcounter{section}{1}
 In this appendix, we discuss the approach of Fukugawa and Fujitani (2010) 
to the modification of the Eulerian action, needed to include the rotational 
fluid velocity component of the Clebsch potential form for ${\bf u}$ that is 
independent of the entropy gradient. They use the action:
\begin{equation}
{\cal A}_{tot}=\int_V\int_{t_{init}}^{t_{fin}} \rho \ell_m dt\ d^3x 
+\int_V\int_{t_{init}}^{t_{fin}} \rho \beta_s \frac{dA_s}{dt} dt\ d^3x, \label{eq:Aa1}
\end{equation}
where $\ell_m$ is the Lagrangian (\ref{eq:3.16ae}) (Fukugawa and Fujitani (2010)
use $\beta_s\to -\beta_s$). In (\ref{eq:Aa1}) $u^i=\partial x^i/\partial t$ has been used 
in the derivation. The variational path for Lagrange label variations 
and the equations $A_s({\bf x},t)=const_s$ and the inverse equations
 ${\bf x}={\bf x}({\bf a},t)$ 
define the Eulerian position of the path.
The $\beta_s=\beta_s({\bf a})$  and $A^s({\bf a})$ 
are functions of the Lagrange labels ${\bf a}$.
We obtain:
\begin{align}
{\cal A}_{tot}=&\int_V\int_{t_{init}}^{t_{fin}} dt\ d^3x\biggl\{\frac{1}{2}\rho u^2-\varepsilon(\rho,S)
-\rho\Phi({\bf x})+\phi\left[\deriv{\rho}{t}
+\nabla{\bf\cdot}(\rho{\bf u})\right]\nonumber\\
&+\beta\left(\deriv{S}{t}+{\bf u}{\bf\cdot}\nabla S\right)
+\tilde{\lambda}^k\left(\deriv{\mu^k}{t}+{\bf u\cdot}\nabla \mu^k\right)
+\beta_s\rho \frac{dA_s}{dt}\biggr\}. \label{eq:Aa2}
\end{align}
where
\begin{equation}
\phi=-\frac{\tilde{\zeta}}{\rho},\quad\beta=r\rho,\quad 
\tilde{\lambda}^k=\rho\lambda^k. \label{eq:Aa3}
\end{equation}
Note that  $\tilde{\zeta}=-\rho\phi$ ensures that the 
Eulerian mass continuity equation is satisfied. Variation of the action (\ref{eq:A2}), 
gives the Eulerian, Clebsch potential 
equations (e.g. Zakharov and Kuznetsov (1997)). In particular,
\begin{equation}
\frac{\delta {\cal A}_{tot}}{\delta {\bf u}}=\rho {\bf u}+\beta\nabla S
+\tilde{\lambda}^k\nabla\mu^k+\rho \beta_s\nabla A^s-\rho\nabla\phi=0, 
\label{eq:Aa4}
\end{equation}
which gives the Clebsch potential form for ${\bf u}$:
\begin{equation}
{\bf u}=\nabla\phi-\frac{\beta}{\rho}\nabla S-\frac{\tilde{\lambda}^k}{\rho}\nabla\mu^k
-\beta_s\nabla A^s. 
\label{eq:Aa5}
\end{equation}
Similarly we obtain the equations:
\begin{align}
\frac{\delta {\cal A}}{\delta\rho}=&-\left\{\frac{d\phi}{dt}-\left(\frac{1}{2} u^2-w-\Phi({\bf x})\right)\right\}=0, 
\label{eq:Aa6}\\
\frac{\delta {\cal A}}{\delta\phi}=&\deriv{\rho}{t}
+\nabla{\bf\cdot}(\rho{\bf u})=0, \label{eq:Aa7}\\
\frac{\delta A}{\delta S}=&-\left\{\deriv{\beta}{t}
+\nabla{\bf\cdot}(\beta {\bf u})+\rho T\right\}=0, 
\label{eq:Aa8}\\
\frac{\delta {\cal A}}{\delta\beta}=&\deriv{S}{t}+{\bf u\cdot}\nabla S=0, \label{eq:Aa9}\\
\frac{\delta {\cal A}}{\delta\tilde{\lambda}^k}=&\deriv{\mu^k}{t}
+{\bf u}{\bf\cdot}\nabla\mu^k=0, 
\label{eq:Aa10}\\
\frac{\delta {\cal A}}{\delta\mu^k}=&-\left(\deriv{\tilde{\lambda}^k}{t}
+\nabla{\bf\cdot}\left({\bf u}\tilde{\lambda}^k\right)\right)=0. \label{eq:Aa11}\\
\frac{\delta {\cal A}}{\delta\beta_s}=&\rho\left(\deriv{A^s}{t}+{\bf u\cdot\nabla}A^s\right)=0, 
\label{eq:Aa12}\\
\frac{\delta {\cal A}}{\delta A^s}=&-\left(\deriv{(\rho\beta_s)}{t}
+\nabla{\bf\cdot}\left(\rho {\bf u}\beta_s\right)\right)
\equiv-\rho\left(\deriv{\beta_s}{t}
+{\bf u\cdot\nabla}\beta_s\right)=0, \label{eq:Aa13}
\end{align}
Equation (\ref{eq:Aa6}) is Bernoulli's equation. In (\ref{eq:Aa13}) the variational equation 
has been simplified by using the mass continuity equation. 

\section*{Appendix B}
\setcounter{section}{2}
\setcounter{equation}{0}
In this appendix we indicate how the same formalism in Section~5, carried out for the 
case of $n=3$ independent Lagrangian mass coordinates also applies for the case $n=2$. 
In the case $n=2$ the cofactor matrix $A_{ij}$ is much simpler than in the $n=3$ case 
(see (\ref{eq:3.12})). For $n=2$ we label the dependent variables $z^k$ as indicated below:
\begin{align}
&(z^1,z^2)=(x^1,x^2),\quad (z^3,z^4)=(u^1,u^2), \nonumber\\
&(z^5,z^6,z^7,z^8)=(\pi_{11},\pi_{12},\pi_{21},\pi_{22}), \quad (z^9,z^{10})=(S,r)\nonumber\\
&(z^{11},z^{12},z^{13},z^{14})=(x_{11},x_{12},x_{21},x_{22})\label{eq:appan1}
\end{align}
The Euler Lagrange equations (\ref{eq:el17})) for the $n=2$ case reduce to:
\begin{align}
\deriv{u^1}{t}=&-x_{22}\deriv{p}{m_1}+x_{21}\deriv{p}{m_2}-\deriv{\Phi}{x_1}, \nonumber\\
\deriv{u^2}{t}=&x_{12}\deriv{p}{m^1}-x_{11}\deriv{p}{m^2}-\deriv{\Phi}{x^2}, \label{eq:appan1a}
\end{align}
where we have used (\ref{eq:3.12}) for $A_{ij}$.

The fundamental one-forms describing the system are:
\begin{align}
&\omega^0=u^1 dx^1+u^2 dx^2+r dS, \nonumber\\
&\omega^1=\pi_{11} dx^1+\pi_{21}dx^2=p x_{22} dx^1-px_{12} dx^2, \nonumber\\
&\omega^2=\pi_{12}dx^1+\pi_{22}dx^2\equiv -p x_{21}dx^1+p x_{11}dx^2, \label{eq:appan2}
\end{align}
Note that $\omega^k=\pi_{ik} dx^i$, $\pi_{ij}=p A_{ij}$ where $A_{ij}$ is the 
cofactor of $x_{ij}$, which for $n=2$ 
is given by (\ref{eq:3.12}).
Noting that
\begin{equation}
\omega^\alpha=L^\alpha_s dz^s, \quad \alpha=0,1,2, \label{eq:appan3}
\end{equation}
we identify 
\begin{align}
&L^0_1=u^1,\quad L^0_2= u^2,\quad L^0_9=r, \nonumber\\
&L^1_1=\pi_{11}=pA_{11}=p x_{22}, \quad L^1_2=\pi_{21}=p A_{21}=-p x_{12},\nonumber\\
&L^2_1=\pi_{12}=p A_{12}=-p x_{21},\quad L^2_2=\pi_{22}=p x_{11}. \label{eq:appan4}
\end{align}
Taking the exterior derivative of $\omega^0$ gives: 
\begin{equation}
d\omega^0=\frac{1}{2} {\sf K}^0_{ij} dz^i\wedge dz^j=du^1\wedge dx^1+du^2\wedge dx^2+dr\wedge dS. \label{eq:appan5}
\end{equation}
We obtain:
\begin{equation}
{\sf K}^0_{u^i,x^i}=1,\quad {\sf K}^0_{x^i,u^i}=-1,
\quad {\sf K}^0_{r,S}=1,\quad {\sf K}^0_{S,r}=-1. \label{eq:appan6}
\end{equation}
Similarly we obtain: 
\begin{equation}
d\omega^k=d\pi_{ik}\wedge dx^i=\frac{1}{2} {\sf K}^k_{\alpha\beta} dz^\alpha\wedge dz^\beta \quad k=1,2, 
\label{eq:appan7}
\end{equation}
which gives:
\begin{equation}
{\sf K}^k_{\pi_{ik},x^i}=1,\quad {\sf K}^k_{x^i,\pi_{ik}}=-1. \label{eq:appan8}
\end{equation}
Alternatively using the notation (\ref{eq:appan1}), (\ref{eq:appan6}) and (\ref{eq:appan8})  give:
\begin{align}
&{\sf K}^0_{3,1}={\sf K}^0_{4,2}={\sf K}^0_{10,9}=1, \nonumber\\
&{\sf K}^1_{5,1}={\sf K}^1_{7,2}=1,\quad {\sf K}^2_{6,1}={\sf K}^2_{8,2}=1, \label{eq:appan9}
\end{align}
and ${\sf K}^\alpha_{ij}=-{\sf K}^\alpha_{ji}$ gives the other non-zero ${\sf K}^\alpha_{ij}$. 

In the above scheme $i=1,2$ give the Euler momentum equations (\ref{eq:5.15}), $i=3,4$ give 
the Lagrangian map equations (\ref{eq:5.16}) for $\partial {\bf x}/\partial t$. 
For $5\leq i\leq 8$ we obtain (\ref{eq:5.17}) for $\partial x^p/\partial m^q$. For
$i=9$ and $i=10$ we obtain (\ref{eq:5.18}) for $dr/dt$ and $dS/dt$. For $11\leq i\leq 14$ 
we obtain equations (\ref{eq:5.18a}) for $\pi_{pq}$. 
For the case of $n$ independent Lagrangian mass
coordinates, there are $2 n^2+2n+2$ equations in total (i.e. if $n=2$ 
$2n^2+2n+2=14$, but if $n=3$, $2n^2+2n+2=26$). 

The comoving energy equation (\ref{eq:5.25}) for $n=2$ reduces to:
\begin{equation}
\frac{d}{dt}\left[\frac{1}{2} u^2+e(\tau,S)+\Phi({\bf x})\right] 
+\derv{m^1}\left[p(x_{22} u^1-x_{12} u^2)\right]
+\derv{m^2}\left[p(-x_{21}+x_{11} u^2)\right]=0, 
\label{eq:appan10}
\end{equation}
which can also be written in the form:
\begin{equation}
\frac{d}{dt}\left[\frac{1}{2} u^2+e(\tau,S)+\Phi({\bf x})\right]
+\frac{1}{\rho} \derv{x^k}\left( pu^k\right)=0. \label{eq:appan11}
\end{equation}
The pullback conservation law (\ref{eq:5.31}) has the same form for all $n>1$. 

The vorticity-symplecticity law (\ref{eq:5.62}) for $n=2$ reduces to the potential vorticity 
conservation law:
\begin{equation}
\frac{d}{dt}\left(\frac{\omega^z+\partial(r,S)/\partial(x,y)}{\rho}\right)=0, \label{eq:appan12}
\end{equation}
where 
\begin{equation}
\omega^z=\left(\deriv{u^y}{x}-\deriv{u^x}{y}\right),\quad \frac{\partial(r,S)}{\partial(x,y)}
=\left(\nabla r\times\nabla S\right){\bf\cdot}{\bf e}_z. \label{eq:appan13}
\end{equation}
The above analysis assumes planar Cartesian geometry. Analogous results clearly apply for other geometries
(e.g. spherical polar or for cylindrical polar coordinates) with an ignorable coordinate, but in these cases 
it is important to include the metric as part of the variational principle (see e.g. Bridges et al. (2010)
and Webb et al. (2014c) discuss multi-symplectic systems in which the metric plays a role). 
 
In Section~7, the formulation of variational 
principles (Section~7.1) and the differential forms $\{\beta_p\}$ representing the 
equation system (Section~7.2) are written in a general form for arbitrary $n$ (the case $n=1$ 
is not considered as it is described  in Webb (2015)). The above completes our discussion 
of the $n=2$ case.

\appendix
\section*{Appendix C}
\setcounter{section}{3}
In this appendix we indicate how the pullback conservation laws arise from Noether's first theorem, 
corresponding to translation invariance of the action $A=\int L d^3m dt$ under translations 
in $m^\beta$ where ${\bf m}=(t,m^1,m^2,m^3)$. 
From Webb et al. (2014c), the multi-symplectic form of Noether's first theorem implies 
that if the action is invariant under a Lie transformation of the form:
\begin{equation}
m^{'\alpha}=m^\alpha+\epsilon V^{\alpha},\quad z^{'s}=z^s+\epsilon V^{z^s}, \label{eq:A1}
\end{equation}
and the divergence transformation:
\begin{equation}
L'=L+\epsilon D_{\alpha} \Lambda^\alpha, \label{eq:A2}
\end{equation}
where $D_\alpha\equiv D_{m^\alpha}$ is the total partial derivative with respect to $m^\alpha$, 
then the equation system admits the conservation law:
\begin{equation}
D_{\alpha}\left(W^\alpha+V^{\alpha} L+\Lambda^\alpha\right)=0. \label{eq:A3}
\end{equation}
In the present application, 
\begin{equation}
W^\alpha=\hat{V}^{z^s} z^s_{,\alpha},\quad \hat{V}^{z^s}=V^{z^s}- V^{\alpha} D_{\alpha} z^s. 
\label{eq:A4}
\end{equation}
For the fluid relabelling symmetries with 
\begin{equation}
V^{z^s}=0,\quad \Lambda^\alpha=0,\quad V^{\alpha}=\delta^\alpha_\beta, \label{eq:A5}
\end{equation}
corresponding to translation invariance with respect to $m^\beta$, the conservation 
law (\ref{eq:A3}) reduces to the pullback conservation law (\ref{eq:5.19}).  
\appendix
\section*{Appendix D}
\setcounter{section}{4}

In this appendix we verify the conservation law (\ref{eq:5.35}) by using a Clebsch variable 
Eulerian variational principle. 
In general, (\ref{eq:5.35}) can be thought of as a nonlocal conservation law which is related 
to a Clebsch potential description of fluid mechanics (e.g. Zakharov and Kuznetsov (1997), 
Morrison (1998)), in which there is an external gravitational field described by the gravitational 
potential $\Phi({\bf x})$. In this approach, the fluid equations arise from the constrained 
variational principle, in which the action is given by:
\begin{align}
A=&\int\biggl\{\left(\frac{1}{2}\rho u^2-\varepsilon(\rho,S)-\rho\Phi({\bf x})\right)
+\phi\left(\deriv{\rho}{t}+\nabla{\bf\cdot}(\rho {\bf u})\right)\nonumber\\
&+\beta\left(\deriv{S}{t}+{\bf u}{\bf\cdot}\nabla S\right)
+\lambda\left(\deriv{\mu}{t}+{\bf u}{\bf\cdot}\nabla\mu\right)\biggr\}\ d^3xdt,
\label{eq:5.36}
\end{align}
By setting $\delta A/\delta u^i=0$ we obtain the Clebsch potential representation for the fluid velocity
in the form:
\begin{equation}
{\bf u}=\nabla\phi-r\nabla S-\tilde{\lambda}\nabla\mu\quad \hbox{where}\quad r=\frac{\beta}{\rho}\quad\hbox{and}\quad \tilde{\lambda}=\frac{\lambda}{\rho}.  \label{eq:5.37}
\end{equation}
The quantity $\mu$ is  associated with the circulation of ${\bf u}$ in Kelvin's theorem. 
 The Lagrange multipliers $\phi$, $\beta$ and $\lambda$ 
ensure that the mass continuity equation, the entropy advection equation $dS/dt=0$ 
and the Lin constraint equation (Kelvin's theorem), $d\mu/dt=0$ are satisfied. 
By varying the action (\ref{eq:5.36}) we obtain:
\begin{equation}
\frac{\delta A}{\delta\phi}=\deriv{\rho}{t}+\nabla{\bf\cdot}(\rho {\bf u})=0, 
\quad \frac{\delta A}{\delta\beta}=\frac{dS}{dt}=0, \quad
\frac{\delta A}{\delta \lambda}= \frac{d\mu}{dt}=0. \label{eq:5.38}
\end{equation}
The condition $\delta A/\delta\rho=0$, gives Bernoulli's equation:
\begin{equation}
\frac{d\phi}{dt}+w+\Phi-\frac{1}{2}u^2=0. \label{eq:5.39}
\end{equation}
The variational equations $\delta A/\delta S=0$ and $\delta A/\delta\mu=0$ imply:
\begin{equation}
\frac{\delta A}{\delta S}=-\left(\deriv{\beta}{t}+\nabla{\bf\cdot}(\rho {\bf u})+\rho T\right)=0,
\quad \frac{\delta A}{\delta\mu}=-\left(\deriv{\lambda}{t}+\nabla{\bf\cdot} (\lambda {\bf u})\right)=0, 
\label{eq:5.40}
\end{equation}
Note that the variables $r=\beta/\rho$ and $\tilde\lambda=\lambda/\rho$ satisfy the equations:
\begin{equation}
\frac{dr}{dt}+T=0,\quad \frac{d\tilde{\lambda}}{dt}=0. \label{eq:5.41}
\end{equation}
Clebsch variables can be used to cast the fluid dynamics equations in a canonical Hamiltonian form 
(e.g. Zakharov and Kuznetsov (1997), Morrison (1998), Webb et al (2014c)). The variables $(\rho,\phi)$,
$(S,\beta)$ and $(\mu,\lambda)$ are canonically conjugate variables in this development. 
Using (\ref{eq:5.37})-(\ref{eq:5.41}), the conservation law (\ref{eq:5.35}) can be reduced to the form:
\begin{equation}
\derv{t}\left[\rho\left(\deriv{\phi}{m^i}-\tilde{\lambda}\deriv{\mu}{m^i}\right)\right]
+\derv{x^j}\left[\rho u^j\left(\deriv{\phi}{m^i}-\tilde{\lambda}\deriv{\mu}{m^i}\right)
-\rho x_{ji}\frac{d{\phi}}{dt}\right]=0. \label{eq:5.42}
\end{equation}
Equation (\ref{eq:5.42}) further reduces to:
\begin{equation}
\rho\left[\frac{d}{dt}\left(\deriv{\phi}{m^i}\right)-\derv{m^i}\left(\frac{d\phi}{dt}\right)\right]
-\tilde{\lambda} \deriv{\mu}{m^i}\left[\deriv{\rho}{t}+\nabla{\bf\cdot}(\rho {\bf u})\right]
-\rho\tilde{\lambda}\derv{m^i}\left(\frac{d\mu}{dt}\right)
-\rho \deriv{\mu}{m^i}\frac{d\tilde\lambda}{dt}=0. \label{eq:5.43}
\end{equation}
By using the Clebsch equations (\ref{eq:5.37})-(\ref{eq:5.41}) 
one can verify that the left hand-side of (\ref{eq:5.42}) is zero, which verifies (\ref{eq:5.35}). 
The first term in square braces vanishes because $d/dt$ and $\partial/\partial m^i$ commute. 

\ack
GMW acknowledges stimulating discussions of multi-symplectic systems 
 and Noether's theorems  with Darryl Holm. We acknowledge very useful referee reports 
which helped improve the manuscript. GMW is supported in part by NASA grant NNX15A165G. 
SCA is supported in part by an NSERC grant.

\end{document}